\documentclass[%
reprint,
superscriptaddress,
amsmath,amssymb,
aps,
prx,
]{revtex4-2}

\newif\ifmain
\maintrue

\newif\ifincludesi
\includesitrue   %

\ifincludesi
\newcommand{\refappendixthm}[2]{\Cref{#1} of~\Cref{Sec:proofs} of the SI}

\newcommand{\refappendix}[2]{\Cref{#1} of the SI}

\else
\newcommand{\refappendixthm}[2]{\refappendix{#1}{#2}}

\newcommand{\refappendix}[2]{Section~{#2} of the SI}

\fi

\usepackage[version=3]{mhchem}
\usepackage{graphicx} 
\usepackage{lastpage}
\usepackage{float}
\usepackage[english]{babel}
\addto{\captionsenglish}{%
  
}
\usepackage{array}
\usepackage[usenames,dvipsnames]{xcolor}
\usepackage{setspace}
\usepackage{hyperref}
\usepackage{bbm}
\usepackage{cleveref}
\usepackage{physics}

\usepackage{epstopdf}%

\definecolor{cream}{RGB}{222,217,201}

\usepackage{booktabs}

\usepackage{diagbox}

\usepackage{amsthm}
\newtheorem{theorem}{Theorem}
\newtheorem{lemma}{Lemma}
\newtheorem{definition}{Definition}
\newtheorem{corollary}{Corollary}[theorem]

\DeclareMathOperator{\sgn}{sgn}

\newcommand{\citeref}[1]{Ref.~\citenum{#1}}

\newcommand{\func}[2]{#1\left(#2\right)}

\newcommand{\set}[1]{\{{#1}\}}
\newcommand{\card}[1]{\left|{#1}\right|}
\newcommand{\sumcond}[2]{\substack{{#1}, \\ {#2}}}
\newcommand{\bigO}[1]{\mathcal{O} \left(#1 \right)}

\newcommand{\fact}[1]{{#1}!}

\newcommand{\SymN}[1]{\mathcal{S}[1,{#1}]}

\newcommand{\cycleof}[2]{c_{#1}({#2})}

\newcommand{\pos}{{\mathbf{r}}}

\newcommand{\ptclIdx}[2]{#1_{#2}}
\newcommand{\sliceIdx}[2]{#1^{\left(#2\right)}}
\newcommand{\beadIdx}[3]{%
    \sliceIdx{%
        \ptclIdx{#1}{#2}%
    }{#3}%
}

\newcommand{\slicepos}[1]{%
    \sliceIdx{\pos}{#1}%
}

\newcommand{\posbead}[2]{%
    \beadIdx{\pos}{#1}{#2}%
}
\newcommand{\beadpos}[2]{\posbead{#1}{#2}}

\newcommand{\mass}{m}

\newcommand{\physPotForceSymbol}{\mathbf{F}}
\newcommand{\physPotForceBead}[2]{\beadIdx{\physPotForceSymbol}{#1}{#2}}

\newcommand{\springfrequency}{\omega_P}
\newcommand{\springconstant}{\mass \springfrequency^2}
\newcommand{\springenergyprefix}{\frac{1}{2} \springconstant}

\newcommand{\rdiffsquared}[4]{\left(\beadpos{#1}{#2} - \beadpos{#3}{#4}\right)^2}

\newcommand{\repsym}{G}
\newcommand{\rep}[1]{\repsym[{#1}]}

\newcommand{\boltzmann}[1]{e^{-\beta {#1}}}
\newcommand{\genboltzmann}[2]{e^{-#1 #2}}

\newcommand{\Eperm}[1]{E^{#1}}

\newcommand{\Epermorig}[1]{\Eperm{#1}}
\newcommand{\Vtoorig}[1]{V^{[1,{#1}]}}
\newcommand{\Efromtoorig}[2]{E^{[{#1},{#2}]}}
\newcommand{\Enkorig}[2]{\Efromtoorig{{#1}-{#2}+1}{{#1}}}

\newcommand{\unitfmt}[2]{#1\,#2}
\newcommand{\unit}[2]{\unitfmt{#1}{\text{#2}}}
\newcommand{\density}[1]{\unitfmt{#1}{\text{\AA}^{-3}}}

\newcommand{\massunit}[1]{\unitfmt{#1}{\text{u}}}
\newcommand{\temperature}[1]{\unitfmt{#1}{\mathrm{K}}}
\newcommand{\dt}[1]{\unitfmt{#1}{\text{fs}}}
\newcommand{\meV}[1]{\unitfmt{#1}{\mathrm{meV}}}
\newcommand{\angstrom}[1]{\unitfmt{#1}{\text{\AA}}}

\newcommand{\gsfev}[1]{\ev{#1}_{\mathrm{GSF}}}
\newcommand{\pfev}[1]{\ev{#1}_{\mathrm{PF}}}

\newcommand{\gsfweightsym}{w_{\mathrm{GSF}}}

\newcommand{\hw}{\hbar \omega}
\newcommand{\bhw}{\beta \hw}
\newcommand{\bhwP}{\frac{\bhw}{P}}

\newcommand{\StatWeight}[2]{W_{\mathrm{#1}}^{\left(#2\right)}}
\newcommand{\BosonicWeight}[1]{\StatWeight{B}{#1}}
\newcommand{\FermionicWeight}[1]{\StatWeight{F}{#1}}

\newcommand{\EnsAvg}[2]{\ev{#1}_{\mathrm{#2}}}

\newcommand{\MinTrotter}{P^{\ast}}
\newcommand{\MinTrotterPF}{\MinTrotter_{\mathrm{PF}}}
\newcommand{\MinTrotterGSF}{\MinTrotter_{\mathrm{GSF}}}

\newcommand{\RecursiveIntegral}[2]{I_{#1}^{\left(#2\right)}}
\newcommand{\RecursiveIntegralArg}[3]{\func{\RecursiveIntegral{#1}{#2}}{\posbead{#1}{1}, \posbead{#1}{#3}}}
\newcommand{\RecursiveIntegralPerm}[3]{%
    \func{%
        \RecursiveIntegral{#1}{#2}%
    }{%
        \posbead{#1}{1},%
        \posbead{\func{#3}{#1}}{1}%
    }%
}

\newcommand{\uncertainty}[2]{#1 \left(#2\right)}

\begin{document}

\makeatother

\title{Generalized Suzuki-Chin Factorization in Bosonic Path Integral Molecular Dynamics}

\author{Jacob Higer}
\affiliation{School of Physics, Tel Aviv University, Tel Aviv 6997801, Israel.}

\author{Barak Hirshberg}
\email{hirshb@tauex.tau.ac.il}
\affiliation{School of Chemistry, Tel Aviv University, Tel Aviv 6997801, Israel.}
\affiliation{The Ratner Center for Single Molecule Science, Tel Aviv University, Tel Aviv 6997801, Israel.}
\affiliation{The Center for Computational Molecular and Materials Science, Tel Aviv University, Tel Aviv 6997801, Israel.}%

\date{\today} %

\begin{abstract}
Modern implementations of path integral molecular dynamics (PIMD) simulations of distinguishable particles frequently make use of high order factorization schemes for the Boltzmann operator to expedite convergence of equilibrium averages. Among these methods is the generalized Suzuki-Chin factorization (GSF), which is accurate up to fourth order in the imaginary-time step. 
In this work, we show that the GSF decomposition of the Boltzmann operator is applicable to bosonic PIMD, and results in an improved convergence of estimators. In particular, we show that the recently developed quadratic scaling bosonic PIMD need not change when using the GSF. The GSF scheme is implemented as a re-weighting factor for observables, without affecting the sampling generated by the standard, second-order, primitive factorization.
We study the effect of this factorization for bosons in a harmonic trap and a sinusoidal potential. 
We also assess the effectiveness of GSF in calculating fermionic expectation values for harmonically-trapped atoms.
In all these cases, we find that the GSF speeds up convergence with the Trotter number by a factor of $\sim 2-4$ across a wide temperature range, at only a modest computational cost.
\end{abstract}

\maketitle

\section{Introduction}
Path integral molecular dynamics (PIMD) is an important method for studying nuclear quantum effects at finite temperatures~\cite{Parrinello1984,Markland2018}. It is well suited for calculating equilibrium properties of quantum systems, although several generalizations of PIMD have been proposed to study dynamical properties as well~\cite{Manolopoulos2004,Habershon2013,Althorpe2021}. PIMD of distinguishable particles is rooted in the idea that each quantum particle can be mapped to a classical ring polymer, consisting of harmonically coupled \emph{beads} (or \emph{replicas}), with a spring stiffness that is proportional to the temperature of the quantum system~\cite{Ceperley1995}. This isomorphism of partition functions maps a complicated quantum problem into a classical one, which is more amenable to efficient computation. Although the classical system is less computationally demanding than its quantum counterpart, it often requires a large number of beads, making it desirable to find ways of reducing the number of replicas required for convergence.

While performing PIMD simulations of distinguishable particles is nowadays straightforward, taking into account %
exchange symmetry between indistinguishable quantum particles is more challenging, as it requires one to consider exponentially many ring-polymer configurations differing by the way they connect to each other. Recently,~\citet{Hirshberg2019} showed that this combinatorial explosion can be avoided by using an equivalent bosonic potential that obeys a recurrence relation. %
With this approach, the PIMD algorithm for bosonic systems scales only quadratically with the system size~\cite{Feldman2023,Higer2025}, enabling, for the first time, PIMD simulations of thousands of bosons~\cite{Feldman2023}. 

The bosonic PIMD scheme was based on the so-called \emph{primitive factorization} (PF) of the Boltzmann operator, which is accurate up to a second order in the imaginary-time slice $\tau = \beta / P$, where $\beta = (k_B T)^{-1}$ is the inverse temperature and $P$ is the number of beads. %
Both path integral Monte Carlo (PIMC) and distinguishable-particle PIMD frequently employ higher-order factorization schemes to accelerate convergence with respect to the number of imaginary-time slices~\cite{Perez2011,Kapil2016,Jang2001,Yamamoto2005}. 

One of the earliest such schemes, introduced by Takahashi and Imada (TI) in the context of PIMC~\cite{TakahashiImada1984}, required the evaluation of the derivatives of the physical potential. While the TI scheme is known to be effective at reducing the number of beads needed for convergence~\cite{Brualla2004,Yamamoto2005,Lindoy2018}, it also has notable drawbacks. In particular, the derivation of estimators is less straightforward than in other approaches, and the resulting expressions can be rather involved, as exemplified by the estimator for the radial distribution function~\cite{Jang2001}. \citet{Suzuki1995_PhysLettA} and~\citet{Chin1997} proposed a closely related method which, unlike the TI scheme, is based on a genuine factorization of the Boltzmann operator and permits simpler estimators. It was later shown by~\citet{Jang2001} that this method yields a comparable, or in some cases superior, convergence speedup across different model systems. We refer to it as the \emph{generalized Suzuki-Chin factorization} (GSF).

Both the TI and GSF methods have been subsequently implemented in PIMD~\cite{Perez2011,Ceriotti2016,Zhang2025}. What is uniquely challenging about the PIMD implementations is the need in evaluating the forces that stem from the additional terms in the classical ring-polymer Hamiltonian. In particular, they require the evaluation of the Hessian of the potential. As a workaround, \citet{Jang2001} suggested performing simulations in the ``ordinary'' PF ensemble and then re-weighting the observables with appropriate GSF weights.~%
\citet{Perez2011} have later refined the weighting factor and extended the method to estimate spatial and momentum distribution functions. However, when the PF and the GSF distributions have little overlap, re-weighting can become statistically inefficient, especially at low $P$~\cite{Ceriotti2011}. As an alternative,~\citet{Kapil2016} proposed to evaluate the Hessian using finite-differences, allowing for a more ergodic sampling of the phase space at a relatively low computational cost.

These advancements in high order factorization schemes have not been applied to PIMD of identical particles, despite being commonplace in distinguishable-particle PIMD, as well as  PIMC~\cite{Jang2001,Yamamoto2005}.
The purpose of the present work is to integrate the GSF scheme into bosonic PIMD. In this work, we show that this can be done while maintaining the same favorable quadratic scaling as in the original algorithm.

The remainder of this paper is organized as follows. In~\Cref{Sec:background}, we present the required theoretical background, with~\Cref{Sec:background-pimd} focusing on PIMD of distinguishable particles,~\Cref{Sec:background-bosonic-pimd} revisiting its bosonic and fermionic extensions, and~\Cref{Sec:background-gsf} discussing the GSF.~\Cref{Sec:results} presents the results of the paper. In~\Cref{Sec:results-gsf-bosons-theory} we discuss the theoretical reason that enables us to integrate GSF into the quadratic scaling PIMD algorithm. In~\Cref{Sec:exact-qho-discretized-partition-function} we derive an exact, finite-$P$ benchmark to compare our approach against, for bosons in a harmonic trap. Following that, in~\Cref{Sec:results-numerical} we present the numerical results validating the effectiveness of our method for two model systems: the harmonic trap and the sinusoidal field.

\section{Background}
\label{Sec:background}
\subsection{PIMD of distinguishable particles}
\label{Sec:background-pimd}
In computational path integral methods, the continuous Feynman path integral representation of the canonical partition function is discretized into a finite number ($P$) of imaginary-time intervals. This is done by replacing the Boltzmann operator at temperature $T$ with a product of multiple Boltzmann operators at a higher temperature $P T$. Since, in general, the kinetic and potential operators do not commute, the next crucial step is finding an approximation to the high-temperature Boltzmann operator that is both sufficiently accurate and simple enough to evaluate in the position basis.

Different approximations can be devised, and depending on the method, the associated discretization error might vary. 
The simplest and most common approach involves the second order Trotter factorization,
\begin{equation}
\label{Eq:primitive-factorization}
\genboltzmann{
    \tau
}{
    \hat{H}
}
=
\genboltzmann{
    \frac{\tau}{2}
}{
    \hat{V}
}
\genboltzmann{
    \tau
}{
    \hat{T}
}
\genboltzmann{
    \frac{\tau}{2}
}{
    \hat{V}
}
+
\bigO{
    \tau^3
}.
\end{equation}
This primitive factorization is called second order because the overall discretization error in this case is proportional to $P^{-2}$. For a $d$-dimensional system of $N$ distinguishable particles of mass $m$, the PF results in the familiar expression,%
\begin{equation}
\label{Eq:pf-distinguishable-partition-function}
Z 
\propto
\int 
\dd[NP]{\pos} \,
\boltzmann{
    \left(
        E + \bar{U}
    \right)
},
\end{equation}
where $\bar{U}$ is the scaled potential, defined as
\begin{equation}
\label{Eq:scaled-phyiscal-potential}
\bar{U} = \frac{1}{P} \sum_{s=1}^{P} {
    \func{V}{\beadpos{1}{s}, \dots, \beadpos{N}{s}}
},
\end{equation}
and $\slicepos{s} = \beadpos{1}{s}, \dots, \beadpos{N}{s}$ represents the positions of all particles at the imaginary-time slice $s$ whose range is $s=1,\dots,P$. The other quantity, $E$, is the spring energy of the ring polymers,
\begin{equation}
\label{Eq:distinguishable-spring-energy}
E =
\sum_{\ell=1}^{N}
\sum_{s=1}^{P}{
    \frac{1}{2} m \springfrequency^2 \rdiffsquared{\ell}{s}{\ell}{s+1}
},
\end{equation}
with the spring frequency $\omega_P = \sqrt{P} / (\beta \hbar)$ and the cyclic closure condition $\beadpos{\ell}{P+1} = \beadpos{\ell}{1}$. In the limit of $P\to\infty$, the classical partition function of~\Cref{Eq:pf-distinguishable-partition-function} coincides with the partition function of the original quantum system. By redefining the constant prefactor in~\Cref{Eq:pf-distinguishable-partition-function} and introducing fictitious momenta, the multidimensional integral can be sampled using molecular dynamics techniques.
\subsection{PIMD of indistinguishable particles}
\label{Sec:background-bosonic-pimd}
To properly account for bosonic exchange, the discretized partition function must include a sum over all $\fact{N}$ permutations of particle labels at the last imaginary-time slice, such that
\begin{equation}
\label{Eq:bosonic-factorial-partition-function}
Z_B \propto
\int 
\dd[NP]{\pos} \,
\frac{1}{\fact{N}}
\sum_{\sigma}
\boltzmann{
    \left(
        \Epermorig{\sigma} + \bar{U}
    \right)
}.
\end{equation}
The permutation-dependent spring energy {$\Epermorig{\sigma}$} has the same form as~\Cref{Eq:distinguishable-spring-energy}, but with a modified path closure condition, $\beadpos{\ell}{P+1} = \beadpos{\sigma(\ell)}{1}$. Note that \Cref{Eq:bosonic-factorial-partition-function} is based on the primitive factorization, and that only the kinetic part is affected by exchange, since the potential is invariant under particle permutations. Direct sampling of~\Cref{Eq:bosonic-factorial-partition-function} using PIMD is computationally prohibitive, as the number of permutations grows factorially with the system size. Previous works~\cite{Hirshberg2019,Feldman2023} solved this combinatorial explosion by re-writing the partition function of~\Cref{Eq:bosonic-factorial-partition-function} as
\begin{equation}
\label{Eq:bosonic-quadratic-partition-function}
Z_B \propto
\int 
\dd[NP]{\pos} \,
\boltzmann{
    \left(
        \Vtoorig{N} + \bar{U}
    \right)
},
\end{equation}
where the bosonic spring potential $\Vtoorig{N}$ is defined recursively via
\begin{equation}
\label{Eq:forward-potential-recurrence}
\boltzmann{\Vtoorig{N}} 
=
\frac{1}{N} \sum_{k=1}^{N}
\boltzmann{
\left(
    \Vtoorig{N-k} + \Enkorig{N}{k}
\right)
},
\end{equation}
with the initial value $\Vtoorig{0} = 0$. Here, $\Enkorig{N}{k}$ denotes the spring energy associated with a ring-polymer that sequentially connects particles $N - k + 1, \dots, N$, i.e.,
\begin{equation}
\label{Eq:cycle-energy-of-last-k-particles}
\Enkorig{N}{k} = 
\springenergyprefix 
\sum_{\ell=N-k+1}^{N}
\sum_{s=1}^{P}
\rdiffsquared{\ell}{s+1}{\ell}{s}
,
\end{equation}
where $\beadpos{\ell}{P+1}=\beadpos{\ell+1}{1}$ except $\beadpos{N}{P+1}=\beadpos{N-k+1}{1}$.

\Cref{Eq:forward-potential-recurrence} makes it possible to sample the bosonic partition function using PIMD in polynomial time at each step. As was shown by %
~\citet{Feldman2023}, further improvement in the scaling can be achieved by introducing a recurrence relation for the cycle energies $\Enkorig{N}{k}$, extending cycles one particle at a time. Using both recurrence relations in tandem results in a bosonic PIMD algorithm that scales as $\bigO{N^2 + PN}$.

Fermionic expectation values can be estimated from bosonic averages by re-weighting with the corresponding sign~\cite{Troyer2005,Hirshberg2020,Dornheim2020,PhysRevB.111.014521,PhysRevE.107.055308,PhysRevE.110.065303,10.1063/5.0106067,10.1063/5.0171930}. The sign is defined as
\begin{equation}
\label{Eq:fermion-sign}
s = \FermionicWeight{N} / \BosonicWeight{N},
\end{equation}
where $\BosonicWeight{N} = \boltzmann{\Vtoorig{N}}$ is the statistical weight of the bosonic configurations, and $\FermionicWeight{N}$ is the weight associated with the fermionic configurations, which can be calculated recursively using
\begin{equation}
\label{Eq:fermion-weight-recurrence}
\FermionicWeight{N}
= 
\frac{1}{N} 
\sum_{k=1}^{N}
\left(-1\right)^{k-1}
\boltzmann{
    \Enkorig{N}{k}
}
\FermionicWeight{N - k},
\end{equation}
with $\FermionicWeight{0} = 1$. Once $s$ is known, it can be used to calculate the fermionic expectation value of an observable $O$ as
\begin{equation}
\EnsAvg{O}{F}
=
\frac{
    \EnsAvg{O s}{B}
}{
    \EnsAvg{s}{B}
},
\end{equation}
where $\EnsAvg{\dots}{B/F}$ denotes an ensemble average of either the bosonic or the fermionic system, respectively.

\subsection{Generalized Suzuki-Chin factorization}
\label{Sec:background-gsf}
The generalized Suzuki-Chin factorization (GSF) is given by~\cite{Suzuki1995_PhysLettA,Chin1997,Jang2001}
\begin{equation}
\label{Eq:gsf}
    \genboltzmann{
        2 \tau
    }{
        \hat{H}
    }
    =
    \genboltzmann{
        \frac{\tau}{3}
    }{
        \hat{V}_e
    }
    \genboltzmann{
        \tau
    }{
        \hat{T}
    }
    \genboltzmann{
        \frac{4 \tau}{3}
    }{
        \hat{V}_m
    }
    \genboltzmann{
        \tau
    }{
        \hat{T}
    }
    \genboltzmann{
        \frac{\tau}{3}
    }{
        \hat{V}_e
    }
    +
    \bigO{
        \tau^5
    }.
\end{equation}
Here, the effective potential operators $\hat{V}_e$ and $\hat{V}_m$ are defined as
\begin{equation}
\label{Eq:gsf-aux-potentials}
\begin{aligned}
\hat{V}_e 
&= 
\hat{V} 
+ 
\frac{\alpha}{6 m}
\sum_{\ell=1}^N{
    \left(
        \frac{\beta \hbar}{P} \grad_{\ell}{\hat{V}}
    \right)^2
},
\\
\hat{V}_m 
&= 
\hat{V} 
+ 
\frac{1 - \alpha}{12 m}
\sum_{\ell=1}^N{
    \left(
        \frac{\beta \hbar}{P} \grad_{\ell}{\hat{V}}
    \right)^2
},
\end{aligned}
\end{equation}
where the scalar parameter $\alpha$ can be chosen arbitrarily, as long as $0 \leq \alpha \leq 1$. Note that the term appearing in the brackets is proportional to the force acting on particle $\ell$ within the given imaginary-time slice. The GSF is a fourth-order factorization because its overall discretization error is proportional to $P^{-4}$.

Following the convention of~\citet{Perez2011}, the resulting partition function for $N$ distinguishable particles of mass $m$ in $d$ dimensions can be written as
\begin{equation}
\label{Eq:tuckerman-gsf-partition-function}
Z 
\propto
\int 
\dd[NP]{\pos} \,
\boltzmann{
    \left(
        E + \bar{U}
    \right)
}
\func{\gsfweightsym}{\posbead{1}{1},\dots,\posbead{N}{P}}
,
\end{equation}
where the total spring energy $E$ and the scaled potential $\bar{U}$ are the same as in the PF (\Cref{Eq:scaled-phyiscal-potential,Eq:distinguishable-spring-energy}), and $\gsfweightsym$ is the weighting factor associated with the GSF, which is given by
\begin{equation}
\label{Eq:gsf-weighting-factor}
\begin{aligned}
& 
\func{
    \gsfweightsym
}{
    \posbead{1}{1},\dots,\posbead{N}{P}
}
=
\exp\Biggl\{
-\beta \sum_{s=1}^{P/2}
\Biggl[
    \frac{
        \func{V}{\slicepos{2s}} 
        - 
        \func{V}{\slicepos{2s - 1}}
    }{
        3P
    }
\\
&
\phantom{\gsfweightsym}
+ 
\frac{1}{9 m \springfrequency^2 P^2}
\sum_{\ell=1}^{N}
\Biggl(
    \alpha
    \left(
        \physPotForceBead{\ell}{2s-1}
    \right)^2
    +
    (1-\alpha)
    \left(
        \physPotForceBead{\ell}{2s}
    \right)^2
\Biggr)
\Biggr]
\Biggr\}
,
\end{aligned}
\end{equation}
where $\physPotForceBead{\ell}{s} = -\grad_{\ell}{\func{V}{\slicepos{s}}}$ is the classical force acting on bead $s$ of particle $\ell$ due to the physical potential $V$ originating from the quantum Hamiltonian. As noted in~\citeref{Perez2011}, decreasing $P$ reduces the weighting factor, which can increase statistical noise.
Conversely, when $P$ is sufficiently large, $\gsfweightsym$ is close to unity, and the GSF scheme coincides with the PF.

The structure of the Suzuki-Chin factorization in~\Cref{Eq:gsf} implies that odd and even replicas have different statistical weights associated with them. This has an effect on certain estimators, such as the potential energy estimator, where it is preferable to include only the contributions due to the dominant imaginary-time slices, which in our case are the odd replicas. %

A key advantage of expressing the partition function in terms of the GSF weighting factor, is that it allows one to carry out PIMD simulations using the standard PF action, and subsequently obtain expectation values via re-weighting. In this approach, one generates configurations based on the second-order PF action, and evaluates both the observable of interest, $O$, and the weighting factor, $\gsfweightsym$, for the different configurations. The end result is then given by
\begin{equation}
\label{Eq:gsf-observable-as-reweighted-primitive}
\gsfev{O} = 
\frac{
    \pfev{
        O \gsfweightsym
    }
}{
    \pfev{
        \gsfweightsym
    }
},
\end{equation}
where $\ev{\dots}_{\mathrm{PF} / \mathrm{GSF}}$ denotes an ensemble average with respect to either the PF or the GSF action, respectively.

\section{Results}
\label{Sec:results}
\subsection{GSF with bosonic exchange}
\label{Sec:results-gsf-bosons-theory}
In the case of bosons, the canonical density operator acts only within the symmetric subspace of the $N$-particle Hilbert space. %
Because of this, the matrix element of the canonical density operator at the last imaginary-time slice becomes permutation-dependent, $
\mel{
    \posbead{1}{P}, 
    \dots, 
    \posbead{N}{P}
}{
    \genboltzmann{\tau}{\hat{H}}
}{
    \posbead{\sigma(N)}{1}, 
    \dots, 
    \posbead{\sigma(N)}{1}
}
$. In the primitive factorization (\Cref{Eq:primitive-factorization}), the evaluation of this matrix element results in a product of a kinetic term and a potential term. The latter is independent of permutation, because the potential operator, $\hat{V}$, is invariant under permutation of particle labels. Therefore, bosonic exchange is entirely captured by the kinetic term which dictates how the ring-polymers are connected across the neighboring imaginary-time slices.

Fortunately, the same reasoning holds also for the GSF (\Cref{Eq:gsf}). Instead of the bare potential operator $\hat{V}$, the GSF consists of effective potential operators $\hat{V}_e$ and $\hat{V}_m$. 
According to~\Cref{Eq:gsf-aux-potentials}, both potentials consist of $\hat{V}$ together with a sum over the forces it generates on all particles within a given imaginary-time slice.
Crucially, in both cases, the sum runs over all the $N$ particles, and each term carries the same weight, making the entire expression invariant under permutation of particle labels. Therefore, just like in the PF, bosonic exchange %
leaves a mark solely on the kinetic part of the density matrix, leading to the partition function %
\begin{equation}
\label{Eq:gsf-bosonic-partition-function}
Z_B \propto
\int 
\dd[NP]{\pos} \,
\frac{1}{\fact{N}}
\sum_{\sigma}
\boltzmann{
    \left(
        \Epermorig{\sigma} + \bar{U}
    \right)
}
\func{\gsfweightsym}{\posbead{1}{1},\dots,\posbead{N}{P}}.
\end{equation}
It is precisely this fact that allows us to seamlessly integrate the GSF scheme into the existing quadratic scaling PIMD algorithm for bosons. This the first key result of this paper.
The implementation of periodic boundary conditions~\cite{Higer2025} is also unaffected by GSF, because they too are encoded only in the kinetic term.

As far as force and estimator calculations are concerned, the changes required by the GSF are the same as in the distinguishable PIMD. In particular, if one opts for a direct calculation of the forces induced by the effective potentials in~\Cref{Eq:gsf-aux-potentials}, these will affect the beads independently from the bosonic forces. Likewise, estimators will have the same form as in the previous quadratic bosonic algorithm, except for the terms that depend on the potential---these will be amended exactly as in the distinguishable-particle algorithm. For instance, the thermodynamic kinetic energy estimator will have the form
\begin{equation}
\label{Eq:gsf-bosonic-td-ke-estimator}
\begin{aligned}
\ev{K} 
& = 
\frac{dPN}{2\beta}
+
\Biggl\langle
    \Vtoorig{N}
    +
    \beta \pdv{
        \Vtoorig{N}
    }{
        \beta
    }
\\
&
+
\frac{1}{9 m \springfrequency^2 P^2}
\sum_{\ell=1}^{N}
\sum_{s=1}^{P / 2}
\Biggl[
    \alpha
    \left(
        \physPotForceBead{\ell}{2s-1}
    \right)^2
    +
    (1-\alpha)
    \left(
        \physPotForceBead{\ell}{2s}
    \right)^2
\Biggr]
\Biggr\rangle
,
\end{aligned}
\end{equation}
which is similar to the distinguishable-particle case, apart from the spring energy expectation term, which in the bosonic case is replaced with the quantity $\Vtoorig{N}
+
\beta \pdv{
    \Vtoorig{N}
}{
    \beta
}$.
The situation is further simplified for the virial and potential energy estimators, whose expressions are unaffected by quantum statistics~\cite{Hirshberg2020}, so that their GSF formulas are identical to those of distinguishable particles. %
Consequently, we calculate the thermal average of the quantum potential energy using the familiar GSF expression
\begin{equation}
\label{Eq:gsf-bosonic-pe-estimator}
\ev{V}
=
\frac{2}{P}
\sum_{s=1}^{P/2}
\ev{
    \func{V}{\slicepos{2s - 1}}
}
.
\end{equation}
Similarly, the GSF expression for the virial kinetic energy estimator is
\begin{equation}
\label{Eq:gsf-bosonic-virial-estimator}
\ev{K_{\mathrm{vir}}}
=
\frac{1}{P}
\sum_{s=1}^{P/2}
\sum_{\ell=1}^{N} 
\ev{
    \beadpos{\ell}{2s - 1}
    \cdot
    \grad_{\ell}{\func{V}{\slicepos{2s - 1}}}
}
.
\end{equation}

So far, we have discussed bosonic expectation values. However, it should be possible, in principle, to evaluate fermionic expectation values by combining GSF re-weighting with sign re-weighting, i.e.,
\begin{equation}
\label{Eq:gsf-fermionic-expectation}
\EnsAvg{O}{GSF,F} 
= 
\frac{
    \EnsAvg{
        O s \gsfweightsym
    }{
        PF,B
    }
}{
    \EnsAvg{
        s \gsfweightsym
    }{
        PF,B
    }
},
\end{equation}
where $s$ is the sign as it was defined in~\Cref{Eq:fermion-sign},
and $\EnsAvg{\dots}{PF,B}$ denotes an average with respect to the bosonic ensemble in the primitive factorization.

\subsection{Exact finite-$P$ partition function for bosons in a harmonic trap}
\label{Sec:exact-qho-discretized-partition-function}
To analyze the performance of PF and GSF, it is useful to have an exact benchmark for each value of $P$. The harmonic trap provides a particularly convenient choice. %
Previous works~\cite{Schweizer1981,TakahashiImada1984,Kamibayashi2016} derived exact expressions for the finite-$P$ partition function of distinguishable particles by direct integration of the full discrete path integral expression. In~\refappendix{Sec:proofs}{II}, we exploit the recursive structure of Gaussian integrals to derive a different, but equivalent expression, more convenient for extension to indistinguishable particles. 

First, we prove that the canonical discretized partition function within the PF, for $N$ distinguishable particles in a $d$-dimensional harmonic trap, is given by
\begin{equation}
\label{Eq:discretized-qho-z-pf-dist}
Z_P^{\mathrm{PF}}
=
\left(
\frac{1}{
    \func{
        F_{2P + 1}
    }{
        \bhwP
    }
    +
    \func{
        F_{2P - 1}
    }{
        \bhwP
    }
    - 2
}
\right)^{\frac{dN}{2}},
\end{equation}
where $\func{F_{n}}{x}$ are the Fibonacci polynomials. We further show that this expression is equivalent to those of Refs.~\citenum{Schweizer1981,TakahashiImada1984}. In the limit of $P\to\infty$,~\Cref{Eq:discretized-qho-z-pf-dist} approaches the exact quantum partition function $Z=\left(\frac{1}{2}\csch\left(\frac{1}{2}\bhw\right)\right)^{dN}$, as expected.
Likewise, for the GSF, we prove that
\begin{equation}
\label{Eq:discretized-qho-z-gsf-dist}
Z_P^{\mathrm{GSF}}
=
\left(
\frac{1}{
    J_{P - 1}
    - 2
}
\right)^{\frac{dN}{2}},
\end{equation}
where $J_n=A_n+B_n$, and $A_n$ and $B_n$ are polynomials defined by the recurrence relations
\begin{equation}
\label{Eq:gsf-recurrence-polynomials}
\begin{cases}
A_{n} = \frac{
    A_{n - 1} C_{n} - 1
}{
    C_{n - 1}
},
\\
B_{n} = B_{n - 1} + k_{n + 1} C_{n - 1},
\\
C_{n} = B_{n - 1} + \left(k_{n + 1} + 1\right) C_{n - 1},
\\
\end{cases}
\end{equation}
with the initial conditions
\begin{equation}
\label{Eq:gsf-recurrence-initial-conditions}
\begin{cases}
A_{1} = 1 + k_2 + k_1 \left(k_2 + 2\right), 
\\
B_{1} = k_2 + 1,
\\
C_{1} = k_2 + 2,
\end{cases}
\end{equation}
where
\begin{equation}
\label{Eq:gsf-recurrence-alternating-constant}
k_n = 
\begin{cases}
\frac{2}{3}\left(\bhwP\right)^{2} 
\left[
    1+\frac{\alpha}{3}\left(\bhwP\right)^{2}
\right]
,
& 
n\in\mathrm{odd},
\\
\frac{4}{3}
\left(\bhwP\right)^{2}
\left[
    1+\frac{1 - \alpha}{6}\left(\bhwP\right)^{2}
\right],
&
n\in\mathrm{even}.
\end{cases}
\end{equation}
Note that the PF case is recovered when $k_n = \left(\bhwP\right)^2$ for all $n$. \Cref{Eq:discretized-qho-z-pf-dist,Eq:discretized-qho-z-gsf-dist} can be used with $\ev{E_P}=-\partial_{\beta} \ln Z_P$ to obtain the finite-$P$ energy of distinguishable particles for both factorization methods.

To the best of our knowledge, no similar finite-$P$ exact expression was published for the partition function of indistinguishable particles. %
In~\refappendixthm{Thm:bosonic-pf-partition-function-as-recursive-integral}{II}, we derive the finite-$P$ partition function of $N$ indistinguishable particles in a $d$-dimensional harmonic trap, and show it is given by the recurrence relation,
\begin{equation}
\label{Eq:harmonic-finite-trotter-partition-function-recurrence}
Z_{P,N}
=
\frac{1}{N}
\sum_{\ell = 1}^{N}
\left(\pm 1\right)^{\ell - 1}
\left[
    2 \func{
        T_{\ell}
    }{
        \frac{J_{P-1}}{2}
    }
    - 2
\right]^{-\frac{d}{2}}
Z_{P,N-\ell}
,
\end{equation}
where $\func{T_n}{x}$ is the $n$th Chebyshev polynomial of the first kind and the initial condition is $Z_{P,0}=1$. The sign depends on the type of particles, with $+1$ corresponding to bosons and $-1$ to fermions. This expression is valid for both PF and GSF, with their distinction encoded solely in the quantity $J_{P-1}$. When $P \to \infty$,~\Cref{Eq:harmonic-finite-trotter-partition-function-recurrence} reduces to the familiar recurrence relation derived in~\citeref{Borrmann1993}.
By taking the logarithmic derivative of~\Cref{Eq:harmonic-finite-trotter-partition-function-recurrence} with respect to $\beta$, one obtains a recurrence relation for the exact energy of indistinguishable particles at every $P$. 

\subsubsection{The optimal value of $\alpha$}
\label{Sec:optimal-alpha}
We introduce a quantitative criterion for the improvement afforded by the GSF. 
The same criterion will be used to analyze the numerical PIMD results in~\Cref{Sec:results-numerical}.
We consider an observable ``converged'' when its relative error falls below $\delta_{\mathrm{tol}} = 0.5\%$, i.e., when $
\abs{E_P - E_{\mathrm{exact}}}
/ 
\abs{E_{\mathrm{exact}}} \leq 0.5 \%
$. We find this threshold to be reasonable, as it offers adequate accuracy without overestimating the GSF speedup. %
The smallest integer $P$ satisfying this condition is denoted by $P^{\ast}_{\mathrm{PF}/\mathrm{GSF}}$, depending on the factorization. In the case of the GSF, only even values of $P$ are allowed.
The \emph{speedup} is then defined as the ratio $P^{\ast}_{\mathrm{PF}}/P^{\ast}_{\mathrm{GSF}}$. 

 With this definition, we %
 proceed to study how the GSF parameter $\alpha$ affects the exact convergence rate for $N=32$ bosons in a three-dimensional harmonic trap.~\Cref{Fig:bosonic-harmonic-optimal-alpha-heatmap} shows $P^{\ast}_{\mathrm{GSF}}$, the number of beads needed to converge $\ev{K}$ for temperatures in the range $\bhw = 1 - 6$. We observe that as the temperature decreases, $\MinTrotterGSF$ becomes more sensitive to changes in $\alpha$. Across all temperatures considered, the optimal value of the tuning parameter is found to be $\alpha = 0$. Moreover, at this value, $\MinTrotterGSF$ varies more slowly with temperature and remains in single digits down to the lowest temperature considered, $\bhw = 6$. %
 Having determined $\alpha = 0$ to be optimal for the harmonic trap, we use this value throughout the paper, including for the sinusoidal potential.

\begin{figure}
\centering
  \includegraphics[width=1.0\linewidth]{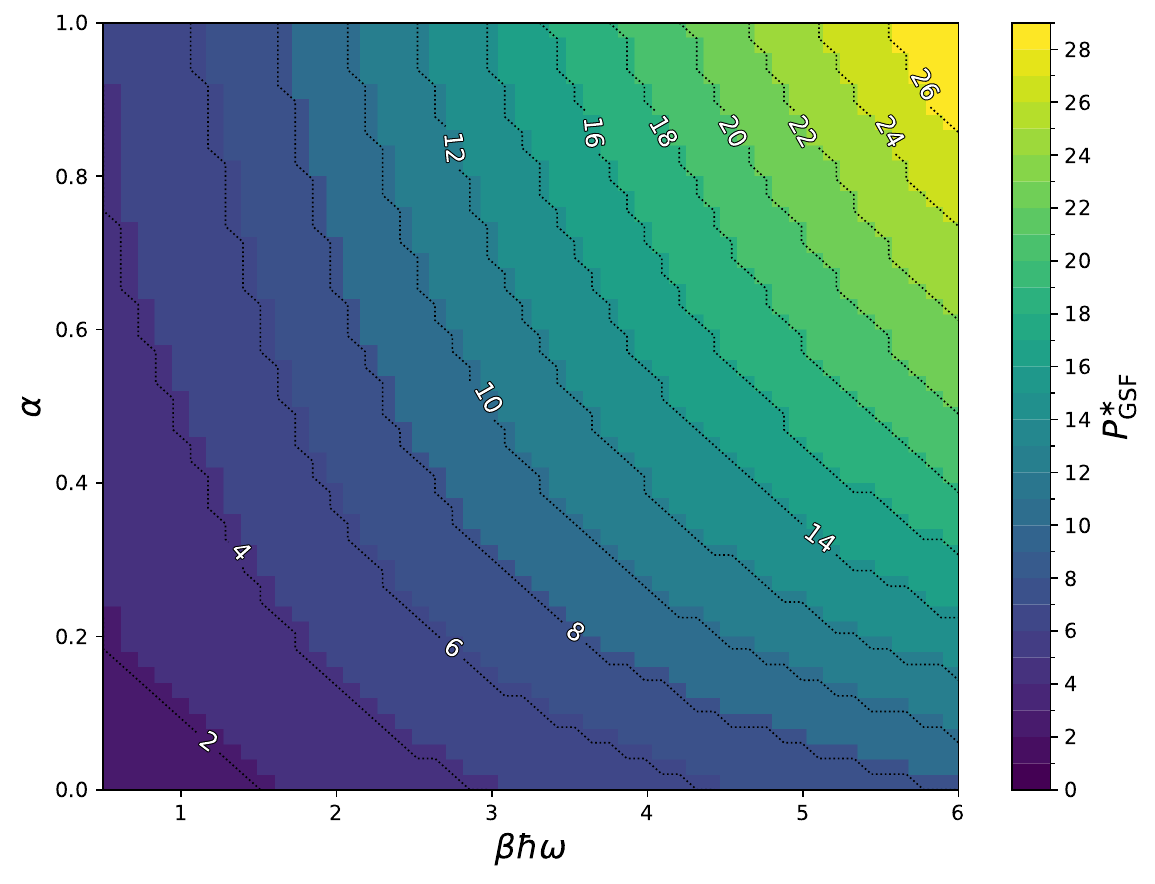}
  \caption{Minimal number of beads $P_{\mathrm{GSF}}^{\ast}$ required to converge the thermodynamic kinetic energy of $N=32$ harmonically-trapped bosons within a fixed tolerance of $0.5\%$, shown as a function of $\bhw$ and the GSF parameter $\alpha$. The color map represents $P_{\mathrm{GSF}}^{\ast}$, while contour lines indicate regions of equal computational cost. The plot illustrates how both the temperature and the choice of $\alpha$ affect the convergence rate of the GSF method.
  }
  \label{Fig:bosonic-harmonic-optimal-alpha-heatmap}
\end{figure}

\begin{figure}
\centering
  \includegraphics[width=1.0\linewidth]{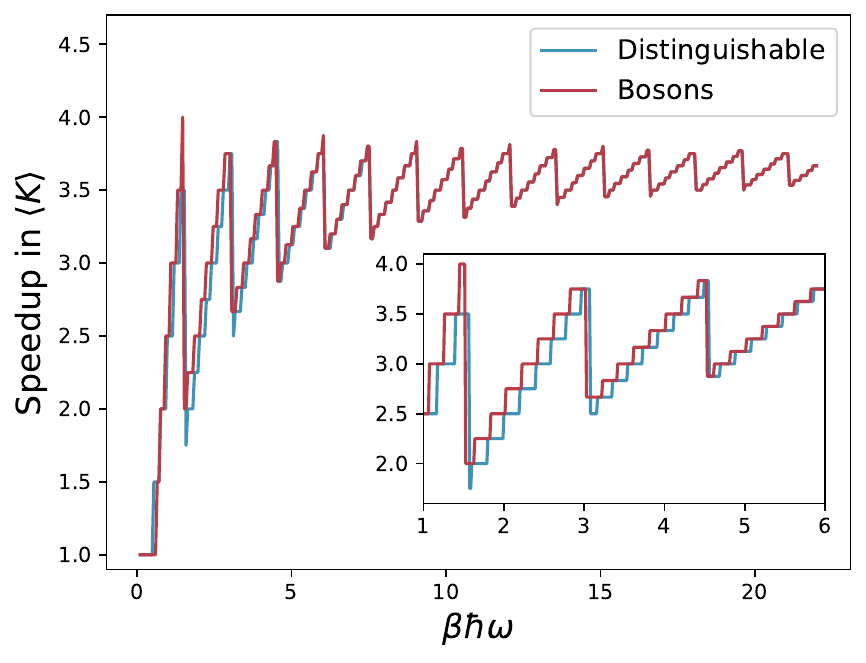}
  \caption{Exact speedup in the convergence of the thermodynamic kinetic energy, $\ev{K}$, as a function of $\bhw$ for $N=32$ harmonically-trapped distinguishable particles (blue) and bosons (red). Inset shows detail in the range $\bhw = 1 - 6$.
  }
  \label{Fig:analytical-harmonic-speedup-vs-bhw}
\end{figure}

\subsubsection{The exact speedup for particles in a harmonic trap}
Once the optimal $\alpha$ has been selected, we compare the exact convergence rate of GSF with that of PF according to the previously defined speedup criterion.~\Cref{Fig:analytical-harmonic-speedup-vs-bhw} shows the exact speedup gained by using the GSF when calculating the kinetic energy, both in the case of bosons and in distinguishable particles. Two notable features emerge. First, for this system and in this temperature regime, GSF is roughly three times faster than PF at the chosen tolerance of $0.5\%$.
Second, the curve displays small, ladder-like steps modulated by larger oscillations. The former are due to the changing $\MinTrotterPF$ as $\bhw$ is varied. By contrast, the speedup decreases whenever the denominator, $\MinTrotterGSF$, increases, which occurs as $\bhw$ increases. Since $\MinTrotterGSF$ remains constant over a significantly wider range of $\bhw$ than $\MinTrotterPF$, the step-like increases are more frequent than the drops. We observe only minor differences in speedup between bosons and distinguishable particles.

In the high-temperature (low $\beta$) limit, the system becomes classical, the PF converges at $P=1$, the GSF has no advantage, and the speedup is unity. %
As the temperature decreases, the speedup rises gradually and ultimately saturates around its maximal value. The weak dependence at large $\bhw$ can be understood by looking at the leading error terms in each factorization. Specifically, the GSF (PF) error is proportional to $\tau^4$ ($\tau^2$), such that when $\bhw$ is large,

\begin{equation}
\mathrm{Speedup}
=
\frac{
    \MinTrotterPF
}{
    \MinTrotterGSF
}
\propto
\frac{
    \beta \delta_{\mathrm{tol}}^{-1/2}
}{
    \beta \delta_{\mathrm{tol}}^{-1/4}
}
=
\delta_{\mathrm{tol}}^{-1/4},
\end{equation}
where $\delta_{\mathrm{tol}}$ is the tolerance. In other words, at high $\bhw$, the speedup should no longer be sensitive to temperature variations, as is the case in~\Cref{Fig:analytical-harmonic-speedup-vs-bhw}.

\subsection{PIMD results}
\label{Sec:results-numerical}
We used the re-weighting method as described in~\Cref{Sec:background-gsf,Sec:results-gsf-bosons-theory} to calculate the quantum energies of bosons and fermions in a harmonic trap, as well as bosons in a sinusoidal potential. Since these systems are non-interacting, exact energies are known~\cite{Borrmann1993,Schmidt2002} and can be compared to the PIMD results.~%
For the harmonic trap, we can also compare PIMD results with an exact benchmark for every $P$, as derived in~\Cref{Sec:exact-qho-discretized-partition-function}.
We analyze the performance of the PF and the GSF with increasing $P$. We then examine the speedup in convergence (as defined in~\Cref{Sec:optimal-alpha}) as a function of temperature to see in which regime the GSF is most beneficial. Following the results in~\Cref{Sec:exact-qho-discretized-partition-function} and~\citeref{Perez2011}, we use $\alpha = 0$ for the GSF in all our simulations.~%
Unless stated otherwise, all simulations use $10^7$ MD steps and averages over $10$-$30$ independent trajectories. If not shown, the statistical error is smaller than the symbol size. Additional computational details are available in~\refappendix{Sec:numerical}{A}.

To reduce the effect of bias and fluctuations, the convergence criterion %
and the subsequent speedup analysis are applied to a continuous fit to the data points. Since observables can be expanded in even powers of the imaginary-time slice~\cite{Fye1986,Suzuki1986,Chin2023_Oct,Chin2023_Dec} $
\tau \propto 1 / P
$, we generate a non-linear least squares fit~\cite{More1978} according to the function
\begin{equation}
\label{Eq:pf-fit-function}
E_{\mathrm{PF}} 
=
E_{\mathrm{exact}}
-
\frac{a_2}{P^2}
+
\frac{a_4}{P^4},
\end{equation}
for PF, and
\begin{equation}
\label{Eq:gsf-fit-function}
E_{\mathrm{GSF}} 
=
E_{\mathrm{exact}}
-
\frac{b_4}{P^4}
+
\frac{b_6}{P^6},
\end{equation}
for GSF. In~\Cref{Eq:pf-fit-function,Eq:gsf-fit-function}, ``$E$'' generically denotes either the total, kinetic or potential energy.

\begin{figure}
\centering
  \includegraphics[width=1.0\linewidth]{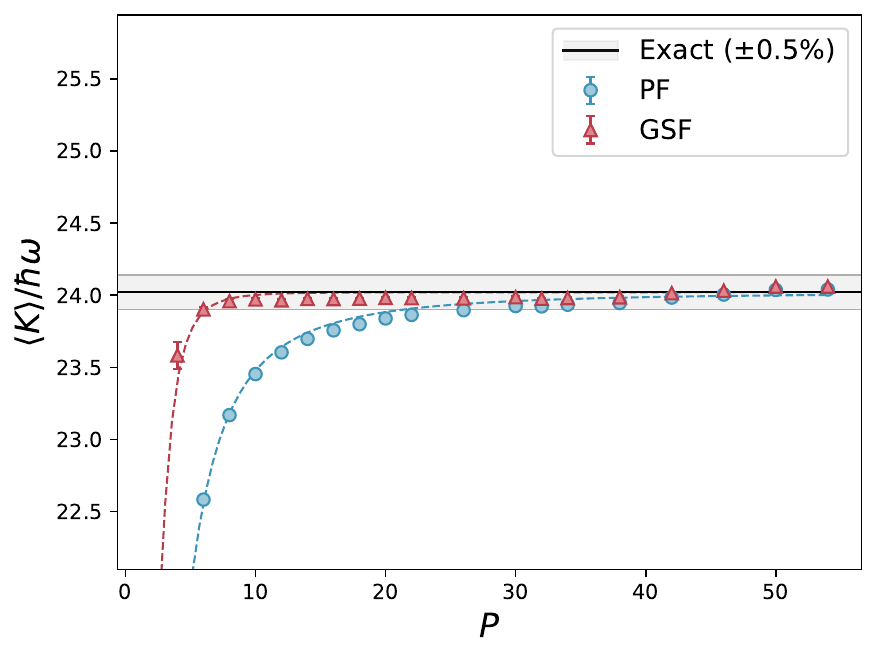}
  \caption{
  Mean kinetic energy (in oscillator units) of $N=32$ bosons in a harmonic trap at $\bhw = 4.33$, calculated using the thermodynamic energy estimator of~\Cref{Eq:gsf-bosonic-td-ke-estimator}, as a function of the number of beads $P$, using the PF (blue circles) and the GSF (red triangles). The dashed lines are fits according to~\Cref{Eq:pf-fit-function,Eq:gsf-fit-function}.
  }
  \label{Fig:bosonic-harmonic-td-prim-energy-convergence}
\end{figure}

In the case of the harmonic trap, we simulate a system of $N=32$ bosons with a trap frequency of $\hw = \meV{3.0}$. Measurements are performed for temperatures in the range $\bhw = 1 - 6$. For each temperature, we examine the behavior of the kinetic and potential energies as a function of the Trotter number $P$.~\Cref{Fig:bosonic-harmonic-td-prim-energy-convergence} depicts the kinetic energy in oscillator units as a function of $P$ for $\bhw = 4.33$. The kinetic energy is evaluated using the primitive estimator of~\Cref{Eq:gsf-bosonic-td-ke-estimator}. Based on our criterion, the PF scheme requires around $P=22$ beads for convergence, as opposed to just $P=8$ beads in the GSF.

In both cases, we see good agreement with the fit, suggesting the correctness of the aforementioned expansion. We find that including the first two terms of the expansion is sufficient %
for all but the highest temperature. In fact, the higher order term ($P^{-4}$ in the PF and $P^{-6}$ in the GSF) has a noticeable contribution only in the high-temperature regime, where convergence is attained at lower values of $P$.

\begin{figure}
\centering
  \includegraphics[width=1.0\linewidth]{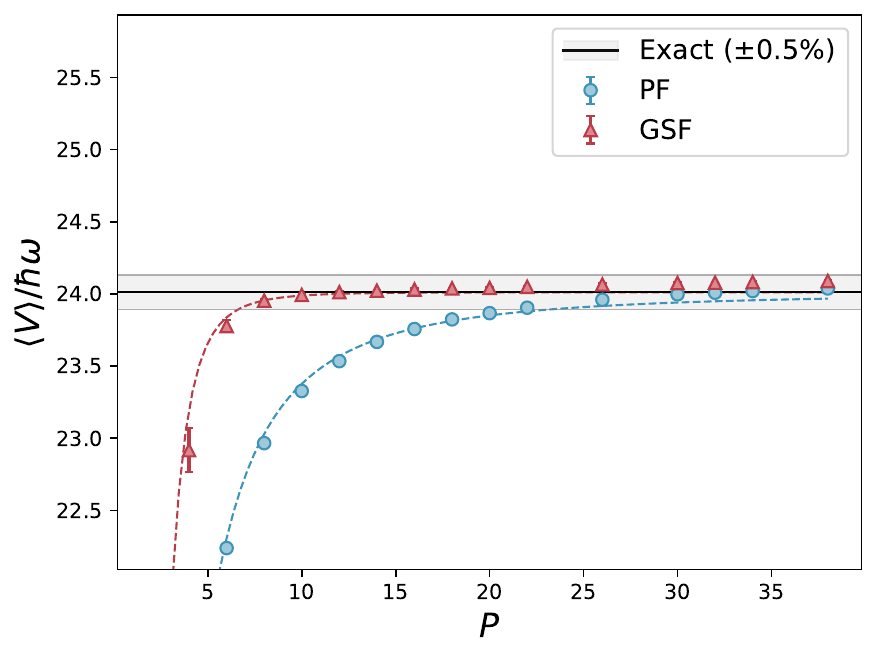}
  \caption{
  Mean potential energy (in oscillator units) of $N=32$ bosons in a harmonic trap at $\bhw = 4.89$, calculated using the potential energy estimator of~\Cref{Eq:gsf-bosonic-pe-estimator}, as a function of the number of beads $P$, using the PF (blue circles) and the GSF (red triangles). The dashed lines are fits according to~\Cref{Eq:pf-fit-function,Eq:gsf-fit-function}.
  }
  \label{Fig:bosonic-harmonic-potential-energy-convergence}
\end{figure}

\Cref{Fig:bosonic-harmonic-potential-energy-convergence} shows a similar analysis for $\bhw = 4.89$, this time examining the convergence of the potential energy, calculated using~\Cref{Eq:gsf-bosonic-pe-estimator}, as a function of $P$. %
Here too, we observe that GSF converges much faster than PF; the latter requires about $P=24$ beads while the former needs only $P=8$ beads. Likewise, we see that the rate of convergence is consistent with~\Cref{Eq:pf-fit-function,Eq:gsf-fit-function}. Both~\Cref{Fig:bosonic-harmonic-td-prim-energy-convergence,Fig:bosonic-harmonic-potential-energy-convergence} demonstrate that at low $P$, the statistical noise increases, which is a known byproduct %
of the re-weighting procedure~\cite{Ceriotti2011}. At the same time, we notice that the error bars quickly diminish in size as $P$ increases, such that the statistical inefficiencies associated with re-weighting have little impact on the overall convergence trend. %

\Cref{Fig:harmonic-kin-energy-convergence-speedup-dist-vs-bosons,Fig:harmonic-pot-energy-convergence-speedup-dist-vs-bosons} show how the speedup in the convergence of $\ev{K}$ and $\ev{V}$, respectively, varies across ten different temperatures in the range $\bhw = 1 - 6$, for both bosons and distinguishable particles. The results are compared with the exact benchmark of~\Cref{Sec:exact-qho-discretized-partition-function}. We find good agreement with the theory, both for the kinetic and potential energy. The relatively minor discrepancies can be attributed to the imperfections of the fit, which can skew the convergence trend and affect the values of $\MinTrotterPF$ and $\MinTrotterGSF$.

\begin{figure}
\centering
  \includegraphics[width=1.0\linewidth]{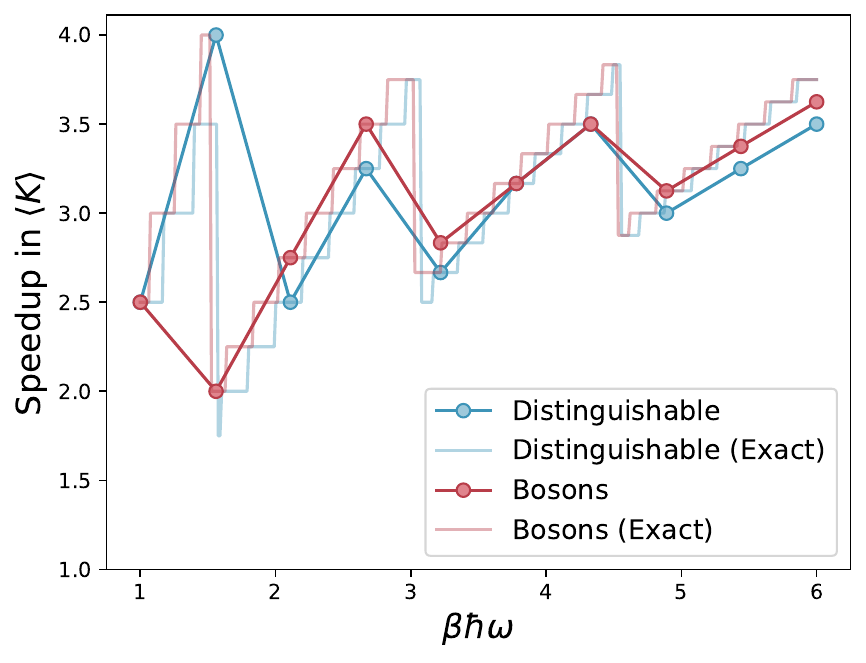}
  \caption{Speedup in the convergence of the thermodynamic kinetic energy estimator, $\ev{K}$, as a function of $\bhw$ for $N=32$ bosons in a harmonic trap with $\hw = \meV{3.0}$. The speedup is determined from the fit to the $\ev{K}$ vs. $P$ plot, for PF and GSF, respectively. The faint solid lines represent the exact speedup calculated from~\Cref{Eq:discretized-qho-z-gsf-dist,Eq:harmonic-finite-trotter-partition-function-recurrence}.
  }
  \label{Fig:harmonic-kin-energy-convergence-speedup-dist-vs-bosons}
\end{figure}

\begin{figure}
\centering
  \includegraphics[width=1.0\linewidth]{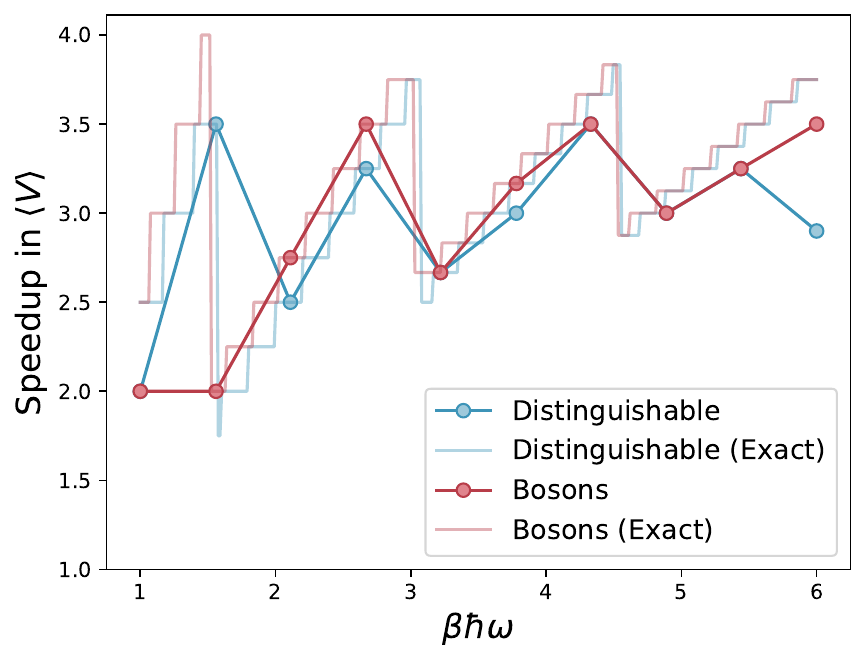}
  \caption{Speedup in the convergence of $\ev{V}$ as a function of $\bhw$ for $N=32$ bosons in a harmonic trap with $\hw = \meV{3.0}$. The speedup is determined from the fit to the $\ev{V}$ vs. $P$ plot, for PF and GSF, respectively. The faint solid lines represent the exact speedup calculated from~\Cref{Eq:discretized-qho-z-gsf-dist,Eq:harmonic-finite-trotter-partition-function-recurrence}.
  }
  \label{Fig:harmonic-pot-energy-convergence-speedup-dist-vs-bosons}
\end{figure}

To test the correctness of the GSF re-weighting in the fermionic case, we simulate $N=4$ fermions in a harmonic trap, with the same frequency of $\hw = \meV{3.0}$ as before. We choose the temperature $\bhw = 1.5$, where the sign problem is not too severe, and where the fermionic energy is approximately $1.47$ times larger than that of bosons. %
To reduce the variance stemming from the fermionic sign problem, we increase the number of MD steps to $10^8$ and average over $50$ independent trajectories.~\Cref{Fig:fermionic-harmonic-kinetic-virial-energy-convergence} shows the kinetic energy as a function of $P$, calculated using~\Cref{Eq:gsf-bosonic-virial-estimator} and the re-weighting protocol of~\Cref{Eq:gsf-fermionic-expectation}. For the PF, we used the standard sign re-weighting, i.e., $\EnsAvg{K}{PF, F} = \EnsAvg{K s}{PF, B} / \EnsAvg{s}{PF, B}$. Here too, we see that the GSF result converges much faster, requiring only $P=2$ as opposed to $P=7$ of the PF, implying a speedup of $3.5$. 

For this temperature, the average sign is $\EnsAvg{s}{PF,B} \approx \uncertainty{0.0569}{2}$. 
Encouragingly, no discernible increase in error bars is observed when using GSF compared to PF.
According to~\Cref{Fig:fermionic-harmonic-kinetic-virial-energy-convergence}, the total weight associated with the fermionic GSF scheme, $\EnsAvg{s \gsfweightsym}{PF, B}$, is comparable to $\EnsAvg{s}{PF, B}$ over the entire range of Trotter numbers. Indeed, the lowest value of $\EnsAvg{\gsfweightsym}{PF, F}=\EnsAvg{s \gsfweightsym}{PF, B}/\EnsAvg{s}{PF, B}$ is attained at $P=2$, where it is equal to $\uncertainty{0.6908}{4}$, meaning the effective weight of the combined method decreases by at most $31\%$ compared to the primitive fermionic re-weighting. This suggests that GSF can be safely used whenever the fermionic sign is tolerable.

\begin{figure}
\centering
  \includegraphics[width=1.0\linewidth]{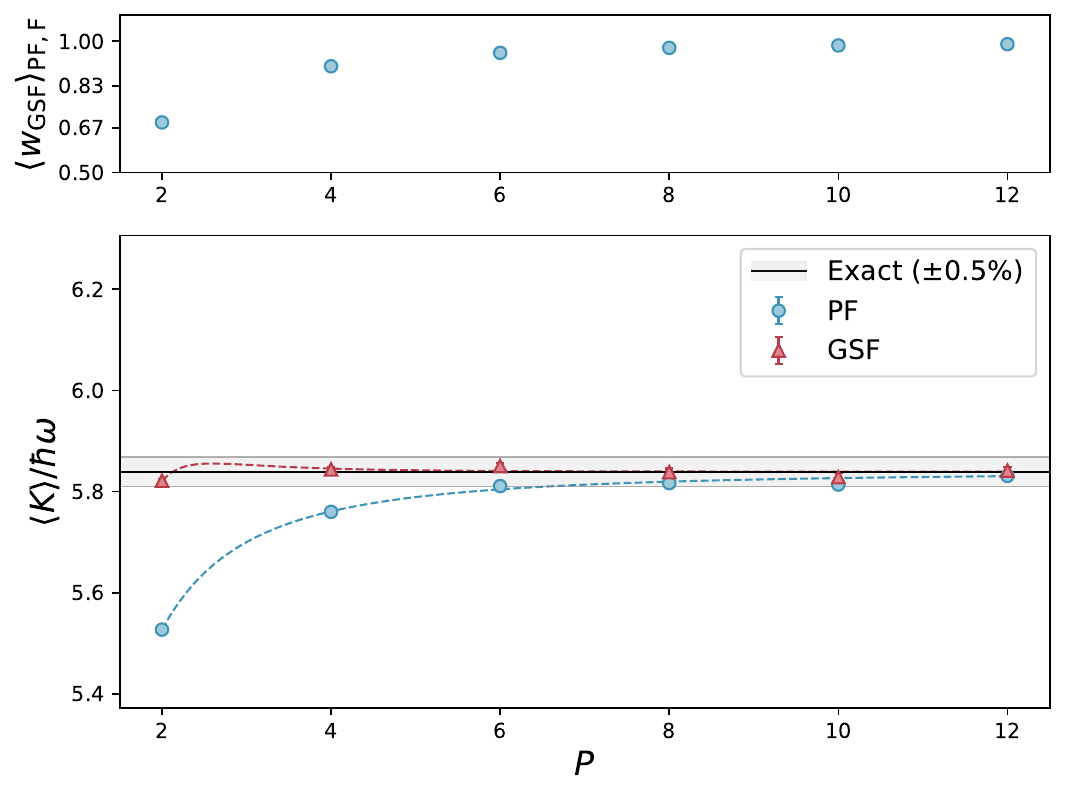}
  \caption{
  \textbf{Bottom panel}: Mean kinetic energy (in oscillator units) of $N=4$ fermions in a harmonic trap at $\bhw = 1.5$, calculated using the virial energy estimator of~\Cref{Eq:gsf-bosonic-virial-estimator}, as a function of the number of beads $P$, using the PF (blue circles) and the GSF (red triangles). The dashed lines are fits according to~\Cref{Eq:pf-fit-function,Eq:gsf-fit-function}. \textbf{Top panel}: Sign-weighted $\gsfweightsym$, defined as $\EnsAvg{\gsfweightsym}{PF, F} = \EnsAvg{s \gsfweightsym}{PF, B} / \EnsAvg{s}{PF, B}$.
  }
  \label{Fig:fermionic-harmonic-kinetic-virial-energy-convergence}
\end{figure}

For the sinusoidal trap, we simulate $N=64$ bosons in the external potential $\func{V}{x,y,z} = V_0 \left(\cos \left(\frac{2\pi}{L} x\right) + \cos \left(\frac{2\pi}{L} y\right) + \cos \left(\frac{2\pi}{L} z\right) \right)$, where $V_0 = \meV{0.3}$ is the oscillation amplitude and $L = \angstrom{12.23}$ is the linear size of the cubic cell. The latter corresponds to a fixed number density of $\density{0.035}$. We impose periodic boundary conditions and apply minimum-image convention to the springs~\cite{Higer2025}. %
Simulations are performed at five temperatures in the range $T=\temperature{0.6-2.5}$, where the mean energy of bosons differs markedly from that of distinguishable particles.~\Cref{Fig:bosonic-cosine-td-prim-energy-convergence} shows the mean energy per atom per trap oscillation amplitude at $T=\temperature{0.6}$. Periodic boundary conditions render the virial estimator invalid, so the energy is calculated by summing~\Cref{Eq:gsf-bosonic-td-ke-estimator} and~\Cref{Eq:gsf-bosonic-pe-estimator}. We observe a convergence trend, this time from above, implying that $a_2$ and $b_4$ from~\Cref{Eq:pf-fit-function,Eq:gsf-fit-function}, respectively, are both negative. With GSF, convergence is achieved already at $P=6$, compared to $P=21$ with PF.%

\begin{figure}
\centering
  \includegraphics[width=1.0\linewidth]{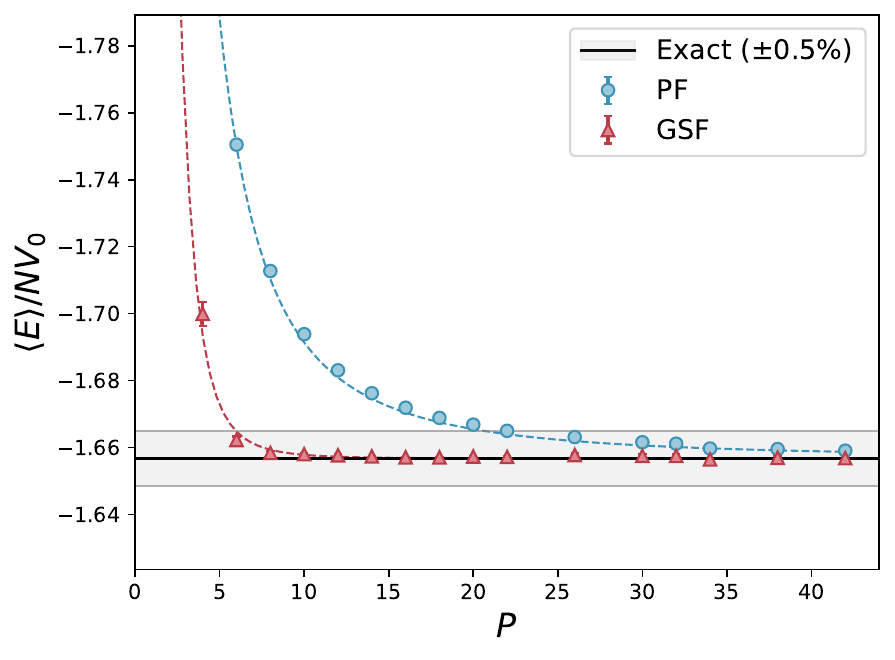}
  \caption{
  Internal energy $\ev{E}/NV_0$ (per trap oscillation amplitude per number of bosons) of $N=64$ bosons in a sinusoidal trap at $T=\temperature{0.6}$, as a function of the number of beads $P$, using the PF (blue circles) and the GSF (red triangles). The dashed lines are fits according to~\Cref{Eq:pf-fit-function,Eq:gsf-fit-function}.
  }
  \label{Fig:bosonic-cosine-td-prim-energy-convergence}
\end{figure}

\Cref{Fig:cosine-energy-convergence-speedup-dist-vs-bosons} depicts the speedup in convergence as a function of temperature, for $T=\temperature{0.6, 1.0, 1.5, 2.0, 2.5}$. We again observe a threefold speedup on average, with variations comparable to the harmonic system and relatively minor differences between bosons and distinguishable particles.

\begin{figure}
\centering
  \includegraphics[width=1.0\linewidth]{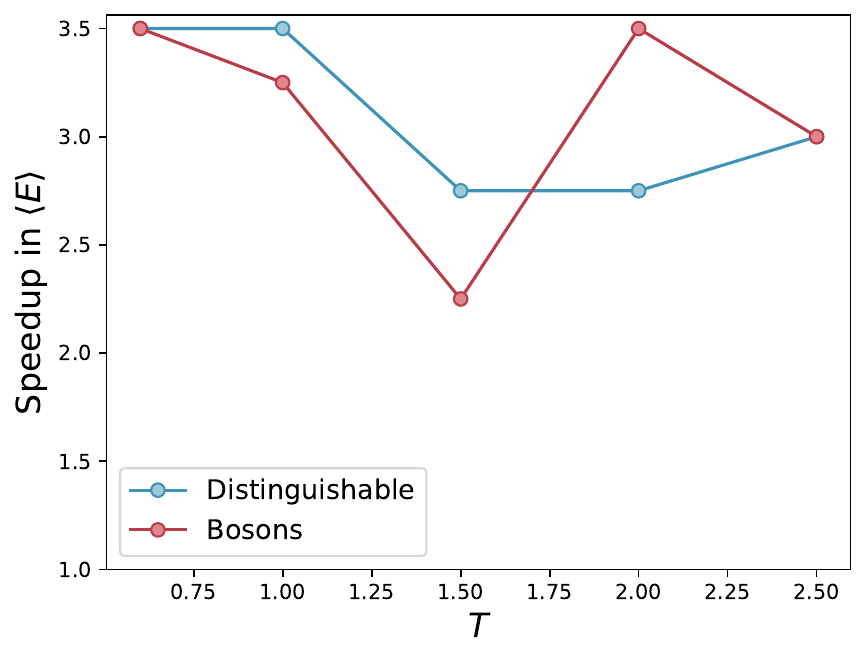}
  \caption{Speedup in the convergence of the thermodynamic energy estimator, $\ev{E}$, as a function of $T$ for $N=64$ bosons in a sinusoidal trap with $A = \meV{0.3}$. The speedup is determined from the fit to the $\ev{E}$ vs. $P$ plot, for PF and GSF, respectively.}
  \label{Fig:cosine-energy-convergence-speedup-dist-vs-bosons}
\end{figure}

\section{Conclusions}
In this paper, we implemented the Suzuki-Chin fourth-order factorization scheme within the framework of the bosonic and fermionic PIMD, without any modifications to the quadratic scaling algorithm. Our argument relied on the fact that, unlike the ring-polymer exchange effects, which are encoded in the kinetic terms, the GSF affects only the potential terms of the Boltzmann operators, and, importantly, keeps them invariant under particle permutations. 

To verify the correctness of our method, we tested it for a system of non-interacting bosons in a harmonic trap and in a sinusoidal potential. %
For identical particles in a harmonic trap, we derived a recurrence relation for the discretized partition function, which enabled an exact comparison of the PF and GSF methods at every $P$, as well as validating the PIMD simulations against these results.

In both cases, GSF significantly reduces the amount of beads necessary for convergence of the energy. Furthermore, using a smaller, harmonically-trapped system, we found that combining GSF re-weighting with sign re-weighting makes it possible to speed up convergence in fermionic systems as well. %
In particular, we find that the weight obtained by combining GSF with fermionic re-weighting is comparable to the fermionic sign, explaining why the GSF does not worsen the sign problem.
Despite the inefficiency inherent to the re-weighting procedure at small Trotter numbers, we found that this did not affect the results in any meaningful way for the systems in question. To understand where the GSF is most beneficial, we calculated the convergence speedup of the different observables as a function of temperature. For both systems we observed a gradual increase in the speedup as the temperature was decreased, suggesting the GSF scheme is most useful in cases where a high number of beads is required for convergence. We also found little difference between the speedup in bosons and distinguishable particles. Furthermore, the results for the sinusoidal field show that no additional complications arise when periodic boundary conditions are introduced within the GSF scheme, making this approach particularly attractive for condensed phase systems.

\section*{Conflicts of interest}

There are no conflicts to declare.

\section*{Data availability}

The data that support the findings of this study, as well as details about the code, are openly available on GitHub, at \url{https://github.com/Hirshberg-Lab/bosonic-gsf-pimd-data}.

\section*{Acknowledgements}
B.H. acknowledges support by the Israel Science Foundation (grants No.\ 1037/22 and 1312/22). The authors thank Yotam M. Y. Feldman for useful comments on the manuscript. 

\bibliography{rsc}
\bibliographystyle{rsc}

\ifincludesi
\clearpage
\onecolumngrid
\section*{Supplementary Information}
\appendix
\makeatletter
\@ifundefined{ifmain}{\newif\ifmain\mainfalse}{}
\makeatother

\ifmain
\else

\documentclass[%
reprint,
superscriptaddress,
amsmath,amssymb,
aps,
prx,
]{revtex4-2}

\usepackage{graphicx}
\usepackage{dcolumn}
\usepackage{bm}

\usepackage[version=3]{mhchem}
\usepackage{xcolor}
\usepackage{amsmath}
\usepackage{amsthm}
\usepackage{physics}
\usepackage[caption=false]{subfig}
\usepackage[ruled,vlined,linesnumbered]{algorithm2e}
\providecommand{\SetAlgoLined}{\SetLine}
\providecommand{\DontPrintSemicolon}{\dontprintsemicolon}
\usepackage{tikz}
\usepackage{booktabs}
\usepackage{xcolor}
\usepackage[export]{adjustbox}
\usepackage{paralist}
\usepackage{hyperref}
\usepackage{cleveref}
\usepackage{float}
\usepackage{diagbox}

\newcommand{\rememberlines}{\xdef\rememberedlines{\number\value{AlgoLine}}}
\newcommand{\resumenumbering}{\setcounter{AlgoLine}{\rememberedlines}}

\makeatletter
\newcommand{\alghidebottomrule}{\renewcommand{\@algocf@post@ruled}}

\newcommand{\alghidetoprule}{\renewcommand{\@algocf@pre@ruled}}
\makeatother

\newtheorem{theorem}{Theorem}
\newtheorem{lemma}{Lemma}
\newtheorem{definition}{Definition}
\newtheorem{corollary}{Corollary}[theorem]

\DeclareMathOperator{\sgn}{sgn}

\begin{document}

\onecolumngrid

\title{Supporting Information: \\
Generalized Suzuki-Chin Factorization in Bosonic Path Integral Molecular Dynamics}

\author{Jacob Higer}
\affiliation {School of Physics, Tel Aviv University, Tel Aviv 6997801, Israel.}
\author{Barak Hirshberg}%
\email{hirshb@tauex.tau.ac.il}%
\affiliation{School of Chemistry, Tel Aviv University, Tel Aviv 6997801, Israel.}%
\affiliation{The Ratner Center for Single Molecule Science, Tel Aviv University, Tel Aviv 6997801, Israel.}
\affiliation{The Center for Computational Molecular and Materials Science, Tel Aviv University, Tel Aviv 6997801, Israel.}%

\maketitle

\onecolumngrid

\fi

\newcommand{\mainref}[2]{%
    \ifmaintextcompiled
        \Cref{#1}%
    \else
        {#2} in the main text%
    \fi
}

\newcommand{\mainrefplain}[2]{%
    \ifmaintextcompiled
        \Cref{#1}%
    \else
        {#2}%
    \fi
}

\newcommand{\maineqrefplain}[2]{\mainrefplain{#1}{Equation (#2)}}

\newcommand{\mainfigref}[2]{\mainref{#1}{Figure #2}}
\newcommand{\maineqref}[2]{\mainref{#1}{Equation (#2)}}
\newcommand{\mainsecref}[2]{\mainref{#1}{Section #2}}
\newcommand{\maineqrefs}[2]{\mainref{#1}{Equations #2}}

\section{Numerical details}
\label{Sec:numerical}

In this section, we provide additional details pertaining to the PIMD simulations.

\subsection{General computational details}

\paragraph{Simulation setup}
All simulations are performed in a three-dimensional cubic geometry with particle mass $\massunit{4.0}$. For harmonically-trapped bosons, we use a simulation cell with linear size of $\angstrom{50}$ and an isotropic harmonic confinement characterized by $\hw = \meV{3.0}$. 

For bosons in the sinusoidal trap, simulations are carried out in a cubic box with side length $L = \angstrom{12.23}$, corresponding to a number density of $\density{0.035}$. The external potential is given by 
\begin{equation}
\func{V}{x,y,z} = V_0 \left(\cos \left(\frac{2\pi}{L} x\right) + \cos \left(\frac{2\pi}{L} y\right) + \cos \left(\frac{2\pi}{L} z\right) \right), 
\end{equation}
with $L$ equal to the box side length and the field amplitude $V_0 = \unit{0.3}{meV}$. We enforce periodic boundary conditions by applying the minimum-image convention to distances between replicas, and by wrapping the bead coordinates~\cite{Higer2025}.

\paragraph{Molecular dynamics protocol} 
Initial bead positions are initialized on a grid (except for the harmonically-trapped bosons, which are sampled uniformly within the simulation cell), while velocities are drawn from a Maxwell-Boltzmann distribution without removing the center of mass motion. Canonical sampling is achieved by applying a Langevin thermostat to the Cartesian degrees of freedom of each bead. The friction coefficient of the thermostat is set to $\left( 100 \Delta t \right)^{-1}$, where $\Delta t$ is the molecular dynamics time step.

\paragraph{Sampling and equilibration}
Observables such as energy are recorded every $100$ MD steps. The time step $\Delta t$ is calibrated using the largest Trotter number $P_{\mathrm{max}}$ for a given temperature, with smaller $P$ requiring proportionally larger time steps. To reduce bias from equilibration, the first $30\%$ of recorded samples are discarded prior to computing ensemble averages.

\subsection{Specific PIMD details}
For harmonically-trapped bosons, we ran simulations at $10$ different temperatures in the range $\bhw = 1 - 6$. For each temperature, we calculated the minimal number of beads required for convergence (with a tolerance of $0.5\%$) of the thermodynamic kinetic energy estimator $\ev{K}$ and the potential energy estimator $\ev{V}$, in both the PF and the GSF. These are denoted as $\func{P^{\ast}_{\mathrm{PF/GSF}}}{O}$, where $O$ is the estimator.

\begin{table}[H]
\centering
\begin{tabular}{ccccccccc}
\toprule 
$\bhw$ 
& 
Time step $\left[\text{fs}\right]$ 
&
$P_{\mathrm{max}}$
& 
Trajectories 
& 
Steps $\left[10^{7}\right]$ 
& 
$\func{P^{\ast}_{\mathrm{PF}}}{K}$ 
& 
$\func{P^{\ast}_{\mathrm{GSF}}}{K}$
& 
$\func{P^{\ast}_{\mathrm{PF}}}{V}$ 
& 
$\func{P^{\ast}_{\mathrm{GSF}}}{V}$
\tabularnewline
\midrule
\midrule 
$1.00$ & $0.020$ & $12$ & $10$ & $1.0$ & $5$ & $2$ & $4$ & $2$
\tabularnewline
\midrule 
$1.56$ & $0.250$ & $10$ & $30$ & $1.0$ 
& $8$ & $4$ & $8$ & $4$
\tabularnewline
\midrule 
$2.11$ & $0.125$ & $28$ & $10$ & $1.0$ 
& $11$ & $4$ & $11$ & $4$
\tabularnewline
\midrule 
$2.67$ & $0.125$ & $34$ & $10$ & $1.0$ 
& $14$ & $4$ & $14$ & $4$
\tabularnewline
\midrule 
$3.22$ & $0.125$ & $42$ & $10$ & $1.0$ 
& $17$ & $6$ & $16$ & $6$
\tabularnewline
\midrule 
$3.78$ & $0.125$ & $50$ & $10$ & $1.0$ 
& $19$ & $6$ & $19$ & $6$
\tabularnewline
\midrule 
$4.33$ & $0.220$ & $32$ & $30$ & $1.0$ 
& $21$ & $6$ & $21$ & $6$
\tabularnewline
\midrule 
$4.89$ & $0.125$ & $64$ & $10$ & $1.0$ 
& $25$ & $8$ & $24$ & $8$
\tabularnewline
\midrule 
$5.44$ & $0.125$ & $70$ & $10$ & $1.0$ 
& $27$ & $8$ & $26$ & $8$
\tabularnewline
\midrule 
$6.00$ & $0.125$ & $78$ & $10$ & $1.0$ 
& $29$ & $8$ & $28$ & $8$
\tabularnewline
\bottomrule
\end{tabular}
\caption{Simulation parameters used for $N=32$ bosons in a harmonic trap, with the corresponding Trotter numbers required for convergence in each factorization.}
\end{table}

For harmonically-trapped fermions, we simulated $N=4$ particles at temperature $\bhw = 1.5$, in the same simulation box as the previous bosonic system. We used a time step of $\dt{1.6}$ for $P_{\mathrm{max}} = 12$, and averaged over $50$ independent trajectories, each spanning $\unitfmt{160}{\mathrm{ns}}$. In the following table, we summarize the results for the sign and GSF weights, as a function of the Trotter number $P$:

\begin{table}[H]
\centering
\begin{tabular}{ccccccc}
\toprule 
\diagbox[width=10em]{
    Weight
}{
    $P$
}
& 
$2$
&
$4$
& 
$6$
& 
$8$
& 
$10$
& 
$12$
\tabularnewline
\midrule
\midrule 
$\EnsAvg{s}{PF,B}$ 
& $\uncertainty{0.0640}{1}$
& $\uncertainty{0.0585}{1}$
& $\uncertainty{0.0574}{2}$
& $\uncertainty{0.0569}{2}$
& $\uncertainty{0.0573}{3}$
& $\uncertainty{0.0571}{2}$
\tabularnewline
\midrule 
$\EnsAvg{sw}{PF,B} / \EnsAvg{s}{PF,B}$
& $\uncertainty{0.6908}{4}$
& $\uncertainty{0.9047}{4}$
& $\uncertainty{0.9557}{3}$
& $\uncertainty{0.9749}{3}$
& $\uncertainty{0.9839}{2}$
& $\uncertainty{0.9887}{2}$
\tabularnewline
\midrule 
$\EnsAvg{sw}{PF,B} / \EnsAvg{w}{PF,B}$ 
& $\uncertainty{0.0570}{1}$
& $\uncertainty{0.0567}{1}$
& $\uncertainty{0.0566}{2}$
& $\uncertainty{0.0565}{2}$
& $\uncertainty{0.0570}{3}$
& $\uncertainty{0.0568}{2}$
\tabularnewline
\bottomrule
\end{tabular}
\caption{Sign and GSF weights for $N=4$ fermions at $\bhw = 1.5$.}
\end{table}

For the sinusoidal trap, we ran simulations at $5$ different temperatures in the range $T=\temperature{0.6 - 2.5}$. Like in the harmonic case, we summarize both the simulation parameters and the corresponding Trotter numbers required for convergence of the thermodynamic energy estimator $\ev{E}$.

\begin{table}[H]
\centering
\begin{tabular}{ccccccc}
\toprule 
Temperature $\left[\text{K}\right]$
& 
Time step $\left[\text{fs}\right]$ 
&
$P_{\mathrm{max}}$
& 
Trajectories 
& 
Steps $\left[10^{7}\right]$ 
& 
$\func{P^{\ast}_{\mathrm{PF}}}{E}$ 
& 
$\func{P^{\ast}_{\mathrm{GSF}}}{E}$
\tabularnewline
\midrule
\midrule 
$0.60$ & $6.0$ & $42$ & $10$ & $1.0$ 
& $48$ & $10$
\tabularnewline
\midrule 
$1.00$ & $5.0$ & $32$ & $10$ & $1.0$ 
& $30$ & $6$
\tabularnewline
\midrule 
$1.50$ & $3.6$ & $28$ & $10$ & $1.0$ 
& $20$ & $4$
\tabularnewline
\midrule 
$2.00$ & $3.6$ & $20$ & $10$ & $1.0$ 
& $16$ & $4$
\tabularnewline
\midrule 
$2.50$ & $3.5$ & $16$ & $10$ & $1.0$ 
& $14$ & $4$
\tabularnewline
\bottomrule
\end{tabular}
\caption{Simulation parameters used for $N=64$ bosons in a sinusoidal trap, with the corresponding Trotter numbers required for convergence in each factorization.}
\end{table}

\subsection{Exact energies of indistinguishable non-interacting systems}

In this work, we compare our PIMD results to exact values for $P \to \infty$. The latter are calculated recursively using~\cite{Krauth2006}
\begin{equation}
\label{eq:noninteracting-energy-recurrence}
\ev{E} 
= 
-\frac{1}{N Z_N} 
\sum_{k=1}^{N}
\left(
    \pm 1
\right)^{k - 1}
\left(
    \pdv{z_k}{\beta} Z_{N-k} + 
    z_k \pdv{Z_{N-k}}{\beta}
\right),
\end{equation}
where $Z_N=\sum_{k=1}^{N} \left(\pm 1\right)^{k-1} z_k Z_{N-k}$ is the partition function of $N$ identical particles (with $+1$ corresponding to bosons and $-1$ to fermions), and $z_k$ is the single-particle partition function at inverse temperature $k\beta$. In the harmonic case, the single-particle quantities have simple closed-form expressions. In the case of the sinusoidal potential, both $z_k$ and $\pdv{z_k}{\beta}$ are evaluated by summing over the first $25$ eigenvalues. Details of the procedure are described in the SI of Ref.~\cite{Higer2025}.

\section{Analytical expressions for harmonic $Z_P$}
\label{Sec:proofs}

In this section, we derive finite-$P$ partition functions for the harmonic potential, in both the PF and the GSF. In~\Cref{Thm:pf-partition-function-as-recursive-integral} we prove the expression for the distinguishable-particle partition function in the PF case. In~\Cref{Cor:pf-dist-discrete-partition-function-same-as-takahashi} we show how this expression coincides with known results from literature, particularly that of~\citeauthor{TakahashiImada1984}. We then generalize these results to the GSF case in~\Cref{Thm:gsf-partition-function-as-recursive-integral}. Finally, in~\Cref{Thm:bosonic-pf-partition-function-as-recursive-integral}, we derive a recurrence relation for the case of indistinguishable particles, valid both for the PF and the GSF.

\begin{lemma}
\label{Lem:gaussian-integral}
Let $a, k > 0$. Then
\begin{align*}
\mathcal{I} =
&
\int \dd{\pos}
\exp{
    - \frac{a}{C}
    \left[
        A \mathbf{x}^2
        - 2 \mathbf{x} \cdot \pos
        + B \pos^2
    \right]
}
\exp{
    - a \left[
        \left(
            \pos
            -
            \mathbf{y}
        \right)^2
        + k \pos^2
    \right]
}
\\
&=
\left(
    \frac{
        \pi
    }{
        \frac{a}{C}
        \left(
            B + C\left(k + 1\right)
        \right)
    }
\right)^{\frac{d}{2}}
\exp[
    -\frac{a A}{C} \mathbf{x}^2
    - a \mathbf{y}^2
    + \frac{a C}{B + C\left(k + 1\right)}
    \left(
        \frac{\mathbf{x}}{C}
        + \mathbf{y}
    \right)^2
].
\end{align*}
\end{lemma}
\begin{proof}[Proof of~\Cref{Lem:gaussian-integral}]
Note that the integrand can be re-written as
\begin{align*}
&
\exp{
    - \frac{a}{C}
    \left[
        A \mathbf{x}^2
        - 2 \mathbf{x} \cdot \pos
        + B \pos^2
    \right]
    - a \left[
        \left(
            \pos
            -
            \mathbf{y}
        \right)^2
        + k \pos^2
    \right]
}
\\
=
&
\exp{
    - \frac{a}{C}
    \left[
        A \mathbf{x}^2
        - 2 \mathbf{x} \cdot \pos
        + B \pos^2
        +
        C
        \left(
            \left(k + 1\right) \pos^2
            - 2 \pos \cdot \mathbf{y}
            + \mathbf{y}^2
        \right)
    \right]
}
\\
=
&
\exp{
    - \frac{a}{C}
    \left(
        A \mathbf{x}^2
        +
        C \mathbf{y}^2
    \right)
}
\exp{
    - \frac{a}{C}
    \left[
        \left(
            B + C \left(k + 1\right)        
        \right)
        \pos^2
        - 2 \left(
            \mathbf{x} + C \mathbf{y}
        \right)
        \cdot
        \pos
    \right]
}.
\end{align*}
The first term is constant with respect to the integral. In the second term, we can complete the square to get
\begin{align*}
&
\exp{
    - \frac{a}{C}
    \left[
        \left(
            B + C \left(k + 1\right)        
        \right)
        \pos^2
        - 2 \left(
            \mathbf{x} + C \mathbf{y}
        \right)
        \cdot
        \pos
    \right]
}
\\
=
&
\exp{
    -\frac{a}{C} 
    \left(
        B + C\left(k + 1\right)
    \right)
    \left(
        \pos
        -
        \frac{
            \mathbf{x} + C\mathbf{y}
        }{
            B + C\left(k + 1\right)
        }
    \right)^2
}
\exp{
    \frac{a}{C}
    \frac{
        \left(
            \mathbf{x} + C \mathbf{y}
        \right)^2
    }{
        B + C\left(k + 1\right)
    }
}.
\end{align*}
Here, only the first term contains the integration variable. This integral is of the form
\begin{equation*}
\int_{\mathbb{R}^d} \dd{\pos} 
e^{-\gamma \left(\pos - \pos_0\right)^2}
=
\left(
    \frac{\pi}{\gamma}
\right)^{d / 2}.
\end{equation*}
In our case, $\gamma = 
\frac{a}{C} 
\left(
    B + C\left(k + 1\right)
\right)$, and so we have
\begin{align*}
\mathcal{I} &=
\left(
    \frac{
        \pi
    }{
        \frac{a}{C}
        \left(
            B + C\left(k + 1\right)
        \right)
    }
\right)^{\frac{d}{2}}
\exp{
    - \frac{a}{C}
    \left(
        A \mathbf{x}^2
        +
        C \mathbf{y}^2
    \right)
}
\exp{
    \frac{a}{C}
    \frac{
        \left(
            \mathbf{x} + C \mathbf{y}
        \right)^2
    }{
        B + C\left(k + 1\right)
    }
}
\\
&=
\left(
    \frac{
        \pi
    }{
        \frac{a}{C}
        \left(
            B + C\left(k + 1\right)
        \right)
    }
\right)^{\frac{d}{2}}
\exp[
    -\frac{a A}{C} \mathbf{x}^2
    - a \mathbf{y}^2
    + \frac{a C}{B + C\left(k + 1\right)}
    \left(
        \frac{\mathbf{x}}{C}
        + \mathbf{y}
    \right)^2
].
\end{align*}.
\end{proof}

\begin{definition}
\label{Def:recursive-integral}
For every $j=1,\dots,P-1$ and a given $i=1,\dots,N$, the quantity $\RecursiveIntegral{i}{j}$ is defined recursively as
\begin{align}
\RecursiveIntegralArg{i}{j}{j+2}
&=
\int \dd{\posbead{i}{j+1}}
\RecursiveIntegral{i}{j - 1}
\exp{
-a \left[
    \left(
        \posbead{i}{j+1}
        -
        \posbead{i}{j+2}
    \right)^2
    + 
    k_{j+1} 
    \left(
        \posbead{i}{j+1}
    \right)^2
\right]
},
\\
\RecursiveIntegralArg{i}{0}{2}
&=
\exp{
-a \left[
    \left(
        \posbead{i}{1}
        -
        \posbead{i}{2}
    \right)^2
    + 
    k_1
    \left(
        \posbead{i}{1}
    \right)^2
\right]
},
\end{align}
where $a,k_j > 0$, and $k_j$ depends on the index $j$.
\end{definition}

\begin{theorem}
\label{Thm:recursive-integral-closed-form}
For every $i=1,\dots,N$ and $j=1,\dots,P-1$, 
\begin{equation}
\RecursiveIntegralArg{i}{j}{j+2}
=
\left(
    \frac{\pi}{a}
\right)^{
    \frac{d j}{2}
}
\left(
    \frac{1}{C_j}
\right)^{
    \frac{d}{2}
}
\exp{
-\frac{a}{C_j}
\left[
    A_j
    \left(
        \posbead{i}{1}
    \right)^2
    -
    2 \posbead{i}{1} \cdot \posbead{i}{j+2}
    +
    B_j
    \left(
        \posbead{i}{j+2}
    \right)^2
\right]
},
\end{equation}
where $A_j, B_j, C_j$ are polynomials in $\set{k_m}_{m=1}^{j+1}$ defined recursively as
\begin{equation}
\label{Eq:general-polynomial-recurrence}
\begin{cases}
A_{j} = \frac{
    A_{j - 1} C_{j} - 1
}{
    C_{j - 1}
},
\\
B_{j} = B_{j - 1} + k_{j + 1} C_{j - 1},
\\
C_{j} = B_{j - 1} + \left(k_{j + 1} + 1\right) C_{j - 1},
\end{cases}
\end{equation}
with the initial conditions
\begin{equation}
\label{Eq:general-polynomial-recurrence-init-cond}
\begin{cases}
A_{1} = 1 + k_2 + k_1 \left(k_2 + 2\right), 
\\
B_{1} = k_2 + 1,
\\
C_{1} = k_2 + 2.
\end{cases}
\end{equation}
\end{theorem}
\begin{proof}[Proof of~\Cref{Thm:recursive-integral-closed-form}]
The proof is by induction on $j$. The statement holds for the base case of $j=1$:
\begin{align*}
\RecursiveIntegralArg{i}{1}{3}
&=
\int \dd{\posbead{i}{2}}
\exp{
-a \left[
    \left(
        \posbead{i}{1}
        -
        \posbead{i}{2}
    \right)^2
    + 
    k_1
    \left(
        \posbead{i}{1}
    \right)^2
\right]
}
\exp{
-a \left[
    \left(
        \posbead{i}{2}
        -
        \posbead{i}{3}
    \right)^2
    + 
    k_{2} 
    \left(
        \posbead{i}{2}
    \right)^2
\right]
}
\\
&=
\left(
    \frac{\pi}{a}
\right)^{\frac{d}{2}}
\left(
    \frac{1}{k_2 + 2}
\right)^{\frac{d}{2}}
\exp{
-\frac{a}{k_2 + 2} 
\left[
    \left(
        1 + k_2 + k_1 \left(k_2 + 2\right)
    \right)
    \left(
        \posbead{i}{1}
    \right)^2
    - 
    2 \posbead{i}{1} \cdot \posbead{i}{3}
    + 
    \left(
        k_{2} + 1
    \right)
    \left(
        \posbead{i}{3}
    \right)^2
\right]
}
\\
&=
\left(
    \frac{\pi}{a}
\right)^{\frac{d}{2}}
\left(
    \frac{1}{C_1}
\right)^{\frac{d}{2}}
\exp{
-\frac{a}{C_1} 
\left[
    A_1
    \left(
        \posbead{i}{1}
    \right)^2
    - 
    2 \posbead{i}{1} \cdot \posbead{i}{3}
    + 
    B_1
    \left(
        \posbead{i}{3}
    \right)^2
\right]
}.
\end{align*}
Assuming the hypothesis is true for $j$, we have, by~\Cref{Def:recursive-integral},
\begin{equation*}
\RecursiveIntegralArg{i}{j+1}{j+3}
=
\int \dd{\posbead{i}{j+2}}
\RecursiveIntegral{i}{j}
\exp{
-a \left[
    \left(
        \posbead{i}{j+2}
        -
        \posbead{i}{j+3}
    \right)^2
    + 
    k_{j+1} 
    \left(
        \posbead{i}{j+2}
    \right)^2
\right]
}.
\end{equation*}
According to the assumption,
\begin{align*}
&\RecursiveIntegralArg{i}{j+1}{j+3}
=
\left(
    \frac{\pi}{a}
\right)^{
    \frac{d j}{2}
}
\left(
    \frac{1}{C_j}
\right)^{
    \frac{d}{2}
}
\int \dd{\posbead{i}{j+2}}
\exp{
-\frac{a}{C_j}
\left[
    A_j
    \left(
        \posbead{i}{1}
    \right)^2
    -
    2 \posbead{i}{1} \cdot \posbead{i}{j+2}
    +
    B_j
    \left(
        \posbead{i}{j+2}
    \right)^2
\right]
}
\\
&
\phantom{\RecursiveIntegralArg{i}{j+1}{j+3}\RecursiveIntegralArg{i}{j+1}{j+3}}
\times
\exp{
-a \left[
    \left(
        \posbead{i}{j+2}
        -
        \posbead{i}{j+3}
    \right)^2
    + 
    k_{j+1} 
    \left(
        \posbead{i}{j+2}
    \right)^2
\right]
}
\\
&=
\left(
    \frac{\pi}{a}
\right)^{
    \frac{d \left(j + 1\right)}{2}
}
\left(
    \frac{1}{B_j + C_{j} \left(k_{j+1} + 1\right)}
\right)^{
    \frac{d}{2}
}
\exp \left\{
-\frac{a}{B_j + C_j \left(k_{j+1}+1\right)} 
\left[
    \frac{
        A_j \left(
            B_j + C_j \left(k_{j+1} + 1\right)
        \right)
        - 1
    }{
        C_j
    }
    \left(
        \posbead{i}{1}
    \right)^2
\right.
\right.
\\
&
\left.
\left.
\phantom{
B_j + C_j \left(k_{j+1}+1\right)
B_j + C_j \left(k_{j+1}+1\right)
B_j + C_j
}
    - 
    2 \posbead{i}{1} \cdot \posbead{i}{j+3}
    + 
    \left(
        B_j + k_{j+1} C_j
    \right)
    \left(
        \posbead{i}{j+3}
    \right)^2
\right]
\right\}
\\
&=
\left(
    \frac{\pi}{a}
\right)^{
    \frac{d \left(j + 1\right)}{2}
}
\left(
    \frac{1}{C_{j+1}}
\right)^{
    \frac{d}{2}
}
\exp{
-\frac{a}{C_{j + 1}}
\left[
    A_{j + 1}
    \left(
        \posbead{i}{1}
    \right)^2
    -
    2 \posbead{i}{1} \cdot \posbead{i}{j+3}
    +
    B_{j + 1}
    \left(
        \posbead{i}{j+3}
    \right)^2
\right]
}
,
\end{align*}
where we used~\Cref{Lem:gaussian-integral} to evaluate the Gaussian integral, and the definitions of the polynomials.
\end{proof}

\begin{theorem}
\label{Thm:fibonacci-from-general-recurrence}
If $k_j = k > 0$ for all $j = 1, \dots, P-1$, then
\begin{align}
\func{A_j}{k} 
&= 
\func{F_{2 j + 3}}{\sqrt{k}}
\\
\func{B_j}{k} 
&= 
\func{F_{2 j + 1}}{\sqrt{k}}
\\
\func{C_j}{k}
&=
\func{U_j}{\frac{k + 2}{2}}
\end{align}
where $\func{F_{n}}{x}$ is the $n$th Fibonacci polynomial and $\func{U_{n}}{x}$ is the $n$th Chebyshev polynomial of the second kind. 
\end{theorem}
\begin{proof}[Proof of~\Cref{Thm:fibonacci-from-general-recurrence}]
When $k_j$ is independent of $j$, the recurrence relations for $B_j$ and $C_j$ (\Cref{Eq:general-polynomial-recurrence}) can be written in matrix form as
\begin{equation*}
\begin{pmatrix}
    B_{j}
    \\
    C_{j}
\end{pmatrix}
=
\begin{pmatrix}
    1 & k
    \\
    1 & k + 1
\end{pmatrix}
\begin{pmatrix}
    B_{j-1}
    \\
    C_{j-1}
\end{pmatrix}
\equiv
M \mathbf{v}_{j-1},
\end{equation*}
where $M$ is the $2 \times 2$ matrix of the coefficients and $\mathbf{v}_j = \begin{pmatrix}
    B_{j}
    \\
    C_{j}
\end{pmatrix}$. It is straightforward to show that the characteristic polynomial of $M$ is $\func{p}{\lambda} = \lambda^2 - \left(k + 2\right) \lambda + 1$, so by the Cayley-Hamilton theorem,
\begin{equation*}
    M^2 - (k + 2) M + \mathbb{I} = 0.
\end{equation*}
Multiplying by $M^{j-1}$ and recognizing that $\mathbf{v}_j = M^{j-1} \mathbf{v}_1$ yields
\begin{equation*}
    \mathbf{v}_{j+2} - \left(k + 2\right) \mathbf{v}_{j+1} + \mathbf{v}_j = 0,
\end{equation*}
which means that $B_j$ and $C_j$ satisfy
\begin{equation}
\label{Eq:decoupled-const-k-recurrence-for-b-c}
\begin{cases}
B_{j+2} = \left(k + 2\right) B_{j+1} - B_j,
\\
C_{j+2} = \left(k + 2\right) C_{j+1} - C_j.
\end{cases}
\end{equation}
From the definition of Fibonacci polynomials, it follows that
\begin{align*}
\func{F_{2j+5}}{x}
&=
x
\func{F_{2j+4}}{x}
+
\func{F_{2j+3}}{x}
\\
&=
x^2 \func{F_{2j+3}}{x}
+
x \func{F_{2j+2}}{x}
+
\func{F_{2j+3}}{x}
\\
&=
\left(x^2 + 2\right)
\func{F_{2j+3}}{x}
-
\func{F_{2j+1}}{x}.
\end{align*}
In other words, $\func{F_{2j+1}}{\sqrt{k}}$ obeys the same recurrence relation as $\func{B_j}{k}$.
Since they coincide for $j=1$ and $j=2$,
\begin{align*}
\func{B_1}{k}
&=
k + 1
=
\func{F_3}{\sqrt{k}},
\\
\func{B_2}{k}
&= 
k^2 + 3k + 1
=
\func{F_5}{\sqrt{k}},
\end{align*}
we conclude that $\func{B_j}{k} = \func{F_{2j + 1}}{\sqrt{k}}$, for all $j > 1$. Similarly,
\begin{align*}
\func{C_1}{k}
&=
k + 2
=
\func{U_1}{\frac{k + 2}{2}},
\\
\func{C_2}{k}
&= 
k^2 + 4k + 3
=
\func{U_2}{\frac{k + 2}{2}},
\end{align*}
and since, by definition, $\func{U_j}{\frac{k + 2}{2}}$ satisfies the same recurrence relation as $\func{C_j}{k}$, we conclude that they coincide for all $j > 1$.

It remains to show that $\func{A_j}{k} = \func{F_{2 j + 3}}{\sqrt{k}} = B_{j + 1}$. First, by directly evaluating the original recurrence, we see that the first two terms do, in fact, coincide:
\begin{align*}
\func{A_1}{k}
&=
k^2 + 3k + 1
=
\func{F_{5}}{\sqrt{k}},
\\
\func{A_2}{k}
&= 
k^3 + 5k^2 + 6k + 1
=
\func{F_{7}}{\sqrt{k}}.
\end{align*}
Next, we must show that $\func{A_j}{k}$ satisfies the same recurrence as
\begin{equation*}
\func{F_{2j+3}}{x} 
= 
\left(x^2 + 2\right) 
\func{F_{2j+1}}{x} 
-
\func{F_{2j-1}}{x}.
\end{equation*}
To this end, we multiply both sides of the equation that defines $A_j$ by $C_{j-1}$ to get
\begin{equation*}
A_{j} C_{j - 1} = A_{j - 1} C_{j} - 1.
\end{equation*}
By replacing $j \to j - 1$ we obtain a second (equivalent) equation. Subtracting the two equations yields
\begin{equation*}
A_{j}C_{j-1}=
A_{j-1}\left(C_{j}+C_{j-2}\right)-A_{j-2}C_{j-1}.
\end{equation*}
Now, using the decoupled recurrence relation for $C_j$ (\Cref{Eq:decoupled-const-k-recurrence-for-b-c}) and dividing by $C_{j-1}$, we obtain
\begin{equation}
A_j = \left(k+2\right) A_{j-1} - A_{j-2},
\end{equation}
which is exactly the same recurrence as that of $\func{F_{2j+3}}{\sqrt{k}}$, concluding our proof.
\end{proof}

\begin{theorem}
\label{Thm:pf-partition-function-as-recursive-integral}
The canonical partition function for $N$ distinguishable particles in a harmonic trap, within the primitive factorization at finite Trotter number $P$, is given by
\begin{equation}
Z^{\mathrm{PF}}_{P} = 
\left(
    \frac{
        a
    }{
        \pi
    }
\right)^{\frac{d N P}{2}}
\int
\prod_{i=1}^{N}
\dd{\posbead{i}{1}}
\RecursiveIntegral{i}{P-1}
=
\left(
\frac{1}{
    \func{
        F_{2P + 1}
    }{
        \sqrt{k}
    }
    +
    \func{
        F_{2P - 1}
    }{
        \sqrt{k}
    }
    - 2
}
\right)^{\frac{d N}{2}},
\end{equation}
where $a=\frac{1}{2}\beta m\omega_{P}^{2}$ and $k=\left(\frac{\bhw}{P}\right)^2$.
\end{theorem}
\begin{proof}[Proof of~\Cref{Thm:pf-partition-function-as-recursive-integral}]
The Trotter-discretized partition function of $N$ distinguishable particles is given by
\begin{equation*}
Z^{\mathrm{PF}}_{P}
= 
\left(
    \frac{
        m P
    }{
        2 \pi \hbar^2 \beta
    }
\right)^{\frac{d N P}{2}}
\int
\dd{\slicepos{1}}
\dots
\dd{\slicepos{P}}
\exp{
    -\beta
    \sum_{j=1}^{P}
    \left[
        \sum_{i=1}^{N}
        \springenergyprefix
        \rdiffsquared{i}{j}{i}{j+1}
        +
        \frac{1}{P}
        \func{V}{\slicepos{j}}
    \right]
}.
\end{equation*}
When the potential is a harmonic trap with angular frequency $\omega$, the partition function $Z_P$ becomes
\begin{align*}
Z^{\mathrm{PF}}_{P}
&= 
\left(
    \frac{
        m P
    }{
        2 \pi \hbar^2 \beta
    }
\right)^{\frac{d N P}{2}}
\int
\dd{\slicepos{1}}
\dots
\dd{\slicepos{P}}
\exp{
    -\frac{1}{2}\beta \springconstant
    \sum_{j=1}^{P}
    \left[
        \sum_{i=1}^{N}
        \rdiffsquared{i}{j}{i}{j+1}
        +
        \frac{1}{P}
        \frac{\omega^2}{\springfrequency^2}
        \left(
            \posbead{i}{j}
        \right)^2
    \right]
}
\\
&=
\left(
    \frac{
        a
    }{
        \pi
    }
\right)^{\frac{d N P}{2}}
\int
\dd{\slicepos{1}}
\dots
\dd{\slicepos{P}}
\exp{
    -a
    \sum_{j=1}^{P}
    \left[
        \sum_{i=1}^{N}
        \rdiffsquared{i}{j}{i}{j+1}
        +
        k
        \left(
            \posbead{i}{j}
        \right)^2
    \right]
}
\\
&= 
\left(
    \frac{
        a
    }{
        \pi
    }
\right)^{\frac{d N P}{2}}
\int
\prod_{i=1}^{N}
\dd{\posbead{i}{1}}
\RecursiveIntegral{i}{P-1}
.
\end{align*}
The last equality follows naturally from~\Cref{Def:recursive-integral} (with $k_j=k$), after unrolling the recurrence $P-1$ times.

In the distinguishable particle case, the integrals can be performed independently from each other. It therefore suffices to calculate $\mathcal{I}=\int
\dd{\posbead{i}{1}}
\RecursiveIntegral{i}{P-1}$ for some $i$. According to~\Cref{Thm:recursive-integral-closed-form},
\begin{equation*}
\mathcal{I} =
\left(
    \frac{\pi}{a}
\right)^{
    \frac{d \left(P - 1\right)}{2}
}
\left(
    \frac{1}{C_{P-1}}
\right)^{
    \frac{d}{2}
}
\int
\dd{\posbead{i}{1}}
\exp{
-\frac{a}{C_{P-1}}
\left[
    A_{P-1}
    \left(
        \posbead{i}{1}
    \right)^2
    -
    2 \posbead{i}{1} \cdot \posbead{i}{1}
    +
    B_{P-1}
    \left(
        \posbead{i}{1}
    \right)^2
\right]
},
\end{equation*}
where we used the path closure condition $\posbead{i}{P+1}=\posbead{i}{1}$. The resulting integral is a Gaussian, which evaluates to
\begin{align*}
\mathcal{I} 
&=
\left(
    \frac{\pi}{a}
\right)^{
    \frac{d \left(P - 1\right)}{2}
}
\left(
    \frac{1}{C_{P-1}}
\right)^{
    \frac{d}{2}
}
\int
\dd{\posbead{i}{1}}
\exp{
-\frac{a}{C_{P-1}}
\left(
    A_{P-1}
    +
    B_{P-1}
    -
    2
\right)
    \left(
        \posbead{i}{1}
    \right)^2
}
\\
&=
\left(
    \frac{\pi}{a}
\right)^{
    \frac{d \left(P - 1\right)}{2}
}
\left(
    \frac{1}{C_{P-1}}
\right)^{
    \frac{d}{2}
}
\left(
    \frac{
        \pi
    }{
        \frac{a}{C_{P-1}}
        \left(
            A_{P-1}
            +
            B_{P-1}
            -
            2
        \right)
    }
\right)^{\frac{d}{2}}
\\
&=
\left(
    \frac{\pi}{a}
\right)^{
    \frac{d P}{2}
}
\left(
    \frac{
        1
    }{
        A_{P-1}
        +
        B_{P-1}
        -
        2
    }
\right)^{\frac{d}{2}}
.
\end{align*}
Using~\Cref{Thm:fibonacci-from-general-recurrence}, we obtain
\begin{equation*}
Z^{\mathrm{PF}}_{P}
=
\left(
    \frac{
        a
    }{
        \pi
    }
\right)^{\frac{d N P}{2}}
\mathcal{I}^N
=
\left(
    \frac{
        1
    }{
        A_{P-1}
        +
        B_{P-1}
        -
        2
    }
\right)^{\frac{dN}{2}}
=
\left(
    \frac{
        1
    }{
        \func{F_{2P+1}}{\sqrt{k}}
        +
        \func{F_{2P-1}}{\sqrt{k}}
        -
        2
    }
\right)^{\frac{dN}{2}}.
\end{equation*}
\end{proof}

\begin{corollary}
\label{Cor:pf-dist-discrete-partition-function-same-as-takahashi}
The representation derived in~\Cref{Thm:pf-partition-function-as-recursive-integral} is equivalent to the formula reported in Refs.~\cite{Schweizer1981,TakahashiImada1984}, which can be expressed as
\begin{equation}
Z^{\mathrm{PF}}_{P} 
=
\left(
    \frac{1}{p^P - q^P}
\right)^{d N},
\end{equation}
where $p,q=\sqrt{\alpha^2 + 1} \pm \alpha$ and $\alpha = \frac{\bhw}{2P}$.
\end{corollary}
\begin{proof}[Proof of~\Cref{Cor:pf-dist-discrete-partition-function-same-as-takahashi}]
Let $x = \bhw / P$. Note that $p$ and $q$ satisfy $pq=1$.
Therefore,
\begin{equation*}
\left(
    p^P - q^P
\right)^2
=
p^{2P} + q^{2P} -2.
\end{equation*}
Defining $r = -q$, we also have
\begin{equation*}
p^{2P} + q^{2P}
=
p^{2P} + r^{2P}.
\end{equation*}
The corollary follows because
\begin{equation*}
p^{2P} + r^{2P}
=
\func{F_{2P + 1}}{x}
+
\func{F_{2P - 1}}{x},
\end{equation*}
where $\func{F_n}{x}$ is the $n$th Fibonacci polynomial. To prove this, note that $L_n = p^n + r^n$ satisfies $L_0 = 2$ and $L_1 = p - q = x$, implying it coincides with Lucas polynomials. The latter are known to obey $\func{L_n}{x} = \func{F_{n+1}}{x} + \func{F_{n-1}}{x}$ (see Ref.~\cite{Koshy2018_Book}), which completes our proof.
\end{proof}

\begin{corollary}
\label{Cor:pf-dist-discrete-partition-function-infinite-p-limit}
In the limit $P\to\infty$, the discretized partition function $Z_{P}^{\mathrm{PF}}$ approaches the exact quantum partition function of a harmonic oscillator.
\end{corollary}
\begin{proof}[Proof of~\Cref{Cor:pf-dist-discrete-partition-function-infinite-p-limit}]
It suffices to show this for a single one-dimensional oscillator. From~\Cref{Cor:pf-dist-discrete-partition-function-same-as-takahashi} follows that
\begin{align*}
\lim_{P \to \infty}
Z_{P}^{\mathrm{PF}}
&=
\lim_{P \to \infty}
\frac{
    1
}{
    \left(
        \sqrt{
            \left( \frac{\bhw}{2P} \right)^2
            + 1
        }
        +
        \frac{\bhw}{2P}
    \right)^P 
    - 
    \left(
        \sqrt{
            \left( \frac{\bhw}{2P} \right)^2
            + 1
        }
        -
        \frac{\bhw}{2P}
    \right)^P 
}
\\
&=
\lim_{P \to \infty}
\frac{
    1
}{
    \left(
        1
        +
        \frac{\bhw}{2P}
    \right)^P 
    - 
    \left(
        1
        -
        \frac{\bhw}{2P}
    \right)^P 
}
\\
&=
\frac{
    1
}{
    e^{\bhw / 2} 
    - 
    e^{-\bhw / 2} 
},
\end{align*}
which is the familiar expression for the canonical partition function of a quantum harmonic oscillator.
\end{proof}

\begin{theorem}
\label{Thm:gsf-partition-function-as-recursive-integral}
The canonical partition function for $N$ distinguishable particles in a harmonic trap, within the GSF at even $P$, is given by
\begin{equation}
Z^{\mathrm{GSF}}_{P} = 
\left(
    \frac{
        a
    }{
        \pi
    }
\right)^{\frac{d N P}{2}}
\int
\prod_{i=1}^{N}
\dd{\posbead{i}{1}}
\RecursiveIntegral{i}{P-1}
=
\left(
\frac{1}{
    A_{P - 1}
    +
    B_{P - 1}
    - 2
}
\right)^{\frac{d N}{2}},
\end{equation}
where $a=\frac{1}{2}\beta m\omega_{P}^{2}$ and
\begin{equation}
k_j = 
\begin{cases}
\frac{2}{3}\left(\bhwP\right)^{2} 
\left[
    1+\frac{\alpha}{3}\left(\bhwP\right)^{2}
\right]
,
& 
j\in\mathrm{odd},
\\
\frac{4}{3}
\left(\bhwP\right)^{2}
\left[
    1+\frac{1 - \alpha}{6}\left(\bhwP\right)^{2}
\right],
&
j\in\mathrm{even}.
\end{cases}
\end{equation}
\end{theorem}
\begin{proof}[Proof of~\Cref{Thm:gsf-partition-function-as-recursive-integral}]
The structure of the proof is similar to that of~\Cref{Thm:pf-partition-function-as-recursive-integral}. The main difference is that the original discretized partition function has a slightly different effective potential, and that we can no longer assume that $k$ is constant in each imaginary-time slice.

In the generalized Suzuki-Chin factorization,
\begin{equation*}
Z^{\mathrm{GSF}}_{P}
= 
\left(
    \frac{
        a
    }{
        \pi
    }
\right)^{\frac{d N P}{2}}
\int
\dd{\slicepos{1}}
\dots
\dd{\slicepos{P}}
\exp{
    -\beta
    \sum_{j=1}^{P}
    \sum_{i=1}^{N}
    \springenergyprefix
    \rdiffsquared{i}{j}{i}{j+1}
    - \beta
    \left[
        \frac{2}{3P}
        \func{V_e}{\slicepos{2j - 1}}
        +
        \frac{4}{3P}
        \func{V_m}{\slicepos{2j}}
    \right]
}.
\end{equation*}
In the case of a harmonic trap,
\begin{equation*}
\grad_{i}{\func{V}{\slicepos{j}}}
=
\frac{1}{2} m \omega^2
\grad_{i}
\left[
    \left(
        \posbead{1}{j}
    \right)^2
    +
    \dots
    +
    \left(
        \posbead{N}{j}
    \right)^2
\right]
=
m \omega^2
\posbead{i}{j}.
\end{equation*}
The effective GSF potentials are
\begin{align*}
\func{V_e}{\slicepos{j}}
&= 
\func{V}{\slicepos{j}}
+
\frac{\alpha}{6m} \left(\frac{\beta\hbar}{P}\right)^2
\sum_{i=1}^{N}
\left(
    \grad_{i}{\func{V}{\slicepos{j}}}
\right)^2
=
\frac{1}{2} m \omega^2
\left[
    1
    +
    \frac{\alpha}{3}
    \left(
        \bhwP
    \right)^2
\right]
\sum_{i=1}^{N}
\left(
    \posbead{i}{j}
\right)^2
,
\\
\func{V_m}{\slicepos{j}}
&= 
\func{V}{\slicepos{j}}
+
\frac{1 - \alpha}{12m} \left(\frac{\beta\hbar}{P}\right)^2
\sum_{i=1}^{N}
\left(
    \grad_{i}{\func{V}{\slicepos{j}}}
\right)^2
=
\frac{1}{2} m \omega^2
\left[
    1
    +
    \frac{1 - \alpha}{6}
    \left(
        \bhwP
    \right)^2
\right]
\sum_{i=1}^{N}
\left(
    \posbead{i}{j}
\right)^2
.
\end{align*}
This implies an effective harmonic potential,
\begin{equation*}
\func{\Phi}{\slicepos{j}}
=
\frac{1}{2} m \Omega_j^2 \sum_{i=1}^{N} \left(\posbead{i}{j}\right)^2,
\end{equation*}
whose frequency depends on the imaginary-time slice parity,
\begin{equation}
\Omega_j^2 =
\omega^2
\begin{cases}
\frac{2}{3P}
\left[
    1+\frac{\alpha}{3}\left(\bhwP\right)^{2}
\right]
,
& 
j\in\mathrm{odd},
\\
\frac{4}{3P}
\left[
    1+\frac{1 - \alpha}{6}\left(\bhwP\right)^{2}
\right],
&
j\in\mathrm{even}.
\end{cases}
\end{equation}
Therefore, the partition function can be re-written as
\begin{equation*}
Z^{\mathrm{GSF}}_{P}
=
\left(
    \frac{
        a
    }{
        \pi
    }
\right)^{\frac{d N P}{2}}
\int
\dd{\slicepos{1}}
\dots
\dd{\slicepos{P}}
\exp{
    -a
    \sum_{j=1}^{P}
    \left[
        \sum_{i=1}^{N}
        \rdiffsquared{i}{j}{i}{j+1}
        +
        k_j
        \left(
            \posbead{i}{j}
        \right)^2
    \right]
},
\end{equation*}
where $a$ is defined as before and $k_j$ is
\begin{equation}
k_j = 
\begin{cases}
\frac{2}{3}\left(\bhwP\right)^{2} 
\left[
    1+\frac{\alpha}{3}\left(\bhwP\right)^{2}
\right]
,
& 
j\in\mathrm{odd},
\\
\frac{4}{3}
\left(\bhwP\right)^{2}
\left[
    1+\frac{1 - \alpha}{6}\left(\bhwP\right)^{2}
\right],
&
j\in\mathrm{even}.
\end{cases}
\end{equation}
The rest of the proof is similar to that of~\Cref{Thm:pf-partition-function-as-recursive-integral}, because~\Cref{Thm:recursive-integral-closed-form} holds for $j$-dependent $k$. However, this also means that we cannot use~\Cref{Thm:fibonacci-from-general-recurrence}. The GSF partition function is therefore
\begin{equation*}
Z^{\mathrm{GSF}}_{P} 
=
\left(
\frac{1}{
    A_{P - 1}
    +
    B_{P - 1}
    - 2
}
\right)^{\frac{d N}{2}},
\end{equation*}
which completes our proof.
\end{proof}

\begin{lemma}
\label{Lem:cyclic-integral-chebyshev}
Let $\mathbf{x}_i \in \mathbb{R}^d$, and $J>2$, $a>0$, $C>0$. Assuming $\mathbf{x}_{N + 1} = \mathbf{x}_1$,
\begin{equation}
\mathcal{I}
=
\int
\dd{\mathbf{x}_1}
\dots
\dd{\mathbf{x}_N}
\exp{
    -\frac{a}{C}
    \sum_{i=1}^{N}
    \left[
        J \left(\mathbf{x}_{i}\right)^{2}
        -2 \mathbf{x}_{i} \cdot \mathbf{x}_{i+1}
    \right]
}
=
\left(
    \frac{\pi C}{a}
\right)^{\frac{dN}{2}}
\left(
    2 \func{T_N}{\frac{J}{2}} - 2
\right)^{-\frac{d}{2}},
\end{equation}
where $\func{T_n}{x}$ is the $n$th Chebyshev polynomial of the first kind.
\end{lemma}
\begin{proof}[Proof of~\Cref{Lem:cyclic-integral-chebyshev}]
Denoting $\mathbf{X} = \left(\mathbf{x}_{1}, \dots, \mathbf{x}_{N}\right)^T$, the integral can be re-written as
\begin{equation*}
\mathcal{I}
=
\int
\dd{\mathbf{x}_1}
\dots
\dd{\mathbf{x}_N}
\exp{
    -\frac{a}{C}
    \mathbf{X}^T M \mathbf{X}
},
\end{equation*}
where $M$ is a positive-definite $dN \times dN$ matrix, which can be expressed as a Kronecker product of a $N \times N$ circulant matrix $Q_N$ and the $d \times d$ identity matrix:
\begin{equation*}
M
=
\begin{pmatrix}
J & -1 & 0 & \dots & -1\\
-1 & J & -1 & \dots & 0\\
0 & -1 & J & \dots & 0\\
\vdots &  &  & \ddots & \vdots\\
-1 & 0 & 0 & \dots & J
\end{pmatrix}_{N \times N}
\otimes
\mathbb{I}_d.
\end{equation*}
In this form, $\mathcal{I}$ is a $dN$-dimensional Gaussian integral, equal to
\begin{equation*}
\mathcal{I}
=
\sqrt{
    \frac{
        \left(\pi C / a\right)^{dN}
    }{
        \det M
    }
}.
\end{equation*}
The determinant of the circulant matrix can be calculated using the standard method of Laplace expansion to minors. Let $\left[Q_N\right]_{i,j}$ denote the $(i,j)$ minor of $Q_N$, and let $D_N$ denote the determinant of a symmetric tridiagonal matrix whose nonzero entries coincide with those of $Q_N$  at the same positions. Then,
\begin{align*}
\det Q_N
&=
J D_{N-1}
+ \left[Q_N\right]_{1,2}
+ \left(-1\right)^{N} \left[Q_N\right]_{1,N}
\\
&=
J D_{N-1}
-
\left(
    D_{N-2}
    +
    1
\right)
-
\left(1 + D_{N-2}\right)
\\
&=
D_{N} - D_{N-2} - 2
\\
&=
2\func{T_N}{\frac{J}{2}}
- 2,
\end{align*}
where we used the fact that $D_N = J D_{N-1} - D_{N-2} = \func{U_N}{J / 2}$, as well as $2 \func{T_N}{x} = \func{U_N}{x} - \func{U_{N-2}}{x}$ (see Ref.~\cite{AbramowitzStegun1965_Book}). Therefore,
\begin{equation*}
\det M = 2^{d} \left[
    \func{T_N}{\frac{J}{2}}
    - 1
\right]^d.
\end{equation*}
Substituting this into the previous expression for $\mathcal{I}$ yields the desired result.
\end{proof}

\begin{theorem}
\label{Thm:bosonic-pf-partition-function-as-recursive-integral}
The canonical partition function for $N$ indistinguishable particles in a harmonic trap, within either the PF or the GSF at finite Trotter number $P$, is given by
\begin{align}
Z_{P,N} 
&= 
\frac{1}{\fact{N}}
\left(
    \frac{
        a
    }{
        \pi
    }
\right)^{\frac{d N P}{2}}
\sum_{\sigma \in \SymN{N}}
\left(
    \pm 1
\right)^{\sigma}
\int
\prod_{i=1}^{N}
\dd{\posbead{i}{1}}
\RecursiveIntegralPerm{i}{P-1}{\sigma}
\\
&=
\frac{1}{N}
\sum_{\ell = 1}^{N}
\left(\pm 1\right)^{\ell - 1}
\left[
    2 \func{
        T_{\ell}
    }{
        \frac{J_{P-1}}{2}
    }
    - 2
\right]^{-\frac{d}{2}}
Z_{P,N - \ell},
\end{align}
with the initial condition $
Z_{P,0} = 1
$. The notation $\left(
    \pm 1
\right)^{\sigma}$ means either $+1$ (bosons) or $\left(-1\right)^{\sigma}\equiv\sgn(\sigma)$ (fermions).
\end{theorem}
\begin{proof}[Proof of~\Cref{Thm:bosonic-pf-partition-function-as-recursive-integral}]
The first equality follows from unwrapping the recurrence for $\RecursiveIntegralPerm{i}{P-1}{\sigma}$. %
For convenience, let us define
\begin{equation*}
\func{I}{\mathbf{x}, \mathbf{y}}
=
\left(
    \frac{\pi}{a}
\right)^{
    \frac{d \left(P - 1\right)}{2}
}
\left(
    \frac{1}{C_{P - 1}}
\right)^{
    \frac{d}{2}
}
\exp{
-\frac{a}{C_{P - 1}}
\left[
    A_{P - 1}
    \mathbf{x}^2
    -
    2 \mathbf{x} \cdot \mathbf{y}
    +
    B_{P - 1}
    \mathbf{y}^2
\right]
}.
\end{equation*}
The proof is then similar to the one provided in Ref.~\citenum{Feldman2023} for the bosonic potential. We write $
Z_{P,N}%
$ as
\begin{equation*}
Z_{P,N} 
=
\frac{1}{\fact{N}}
\left(
    \frac{
        a
    }{
        \pi
    }
\right)^{\frac{d N P}{2}}
\sum_{q=1}^{N}
\sum_{
    \sumcond{
        \sigma \in \SymN{N}
    }{
        \card{\cycleof{\sigma}{N}}=q
    }
}
\left(
    \pm 1
\right)^{
    \rep{\sigma}
}
\int
\dd{
    \posbead{1}{1}
}
\dots
\dd{
    \posbead{N}{1}
}
\prod_{i=1}^{N - q}
\func{I}{
    \posbead{i}{1},
    \posbead{
        \rep{
            \sigma \setminus \cycleof{\sigma}{N}
        }(
            i
        )
    }{
        1
    }
}
\prod_{\ell=N-q+1}^{N}
\func{I}{
    \posbead{\ell}{1},
    \posbead{\ell + 1}{1}
},
\end{equation*}
where $\cycleof{\sigma}{N}$ is the cycle in $\sigma$ that contains $N$, $\rep{\sigma}$ is the representative permutation, and $\posbead{N + 1}{1} = \posbead{N - q + 1}{1}$. The last $q$ integrals can be evaluated independently from the rest, because they belong to a separate cycle. According to~\Cref{Lem:cyclic-integral-chebyshev},
\begin{align*}
&
\int
\prod_{\ell=N-q+1}^{N}
\dd{
    \posbead{\ell}{1}
}
\func{I}{
    \posbead{\ell}{1},
    \posbead{\ell + 1}{1}
}
\\
&=
\left(
    \frac{\pi}{a}
\right)^{
    \frac{
        d q \left(P - 1\right)
    }{
        2
    }
}
\left(
    \frac{1}{C_{P - 1}}
\right)^{
    \frac{dq}{2}
}
\int
\dd{
    \posbead{N-q+1}{1}
}
\dots
\dd{
    \posbead{N}{1}
}
\exp{
    -\frac{a}{C_{P - 1}}
    \sum_{\ell = N - q + 1}^{N}
    \left[
        J_{P - 1}
        \left(
            \posbead{\ell}{1}
        \right)^2
        - 
        2 \posbead{\ell}{1} \cdot \posbead{\ell + 1}{1}
    \right]
}
\\
&=
\left(
    \frac{\pi}{a}
\right)^{
    \frac{
        d q \left(P - 1\right)
    }{
        2
    }
}
\left(
    \frac{1}{C_{P - 1}}
\right)^{
    \frac{dq}{2}
}
\left(
    \frac{
        \pi C_{P-1}
    }{
        a
    }
\right)^{
    \frac{dq}{2}
}
\left(
    2 \func{
        T_q
    }{
        \frac{J_{P-1}}{2}
    }
    - 2
\right)^{
    -\frac{d}{2}
}
\\
&=
\left(
    \frac{\pi}{a}
\right)^{
    \frac{
        d q P
    }{
        2
    }
}
\left(
    2 \func{
        T_q
    }{
        \frac{J_{P-1}}{2}
    }
    - 2
\right)^{
    -\frac{d}{2}
}
,
\end{align*}
where $J_n = A_n + B_n$. (Note that $J_n > 2$ for all $n$ due to the structure of~\Cref{Eq:general-polynomial-recurrence,Eq:general-polynomial-recurrence-init-cond} and the fact that $k_n > 0$.) In the case of fermions, we use the multiplicativity of the sign to write
\begin{equation*}
\left(
    - 1
\right)^{
    \rep{\sigma}
}
=
\left(
    -1
\right)^{
    \rep{
        \sigma \setminus \cycleof{\sigma}{N}
    }
}
\cdot 
\left(
    -1
\right)^{
    \card{\cycleof{\sigma}{N}} - 1
}.
\end{equation*}
Therefore, the partition function is
\begin{equation*}
Z_{P,N} 
=
\frac{1}{\fact{N}}
\left(
    \frac{
        a
    }{
        \pi
    }
\right)^{\frac{d \left(N - q\right) P}{2}}
\sum_{q=1}^{N}
\left(
    \pm 1
\right)^{q - 1}
\left(
    2 \func{
        T_q
    }{
        \frac{J_{P-1}}{2}
    }
    - 2
\right)^{
    -\frac{d}{2}
}
\sum_{
    \sumcond{
        \sigma \in \SymN{N}
    }{
        \card{\cycleof{\sigma}{N}}=q
    }
}
\left(
    \pm 1
\right)^{
    \rep{
        \sigma \setminus \cycleof{\sigma}{N}
    }
}
\int
\prod_{i=1}^{N - q}
\dd{
    \posbead{i}{1}
}
\func{I}{
    \posbead{i}{1},
    \posbead{
        \rep{
            \sigma \setminus \cycleof{\sigma}{N}
        }(
            i
        )
    }{
        1
    }
}.
\end{equation*}
The inner sum over permutations no longer contains the last cycle, and is in fact proportional to $Z_{P,N-q}$. More precisely,
\begin{equation*}
\sum_{
    \sumcond{
        \sigma \in \SymN{N}
    }{
        \card{\cycleof{\sigma}{N}}=q
    }
}
\left(
    \pm 1
\right)^{
    \rep{
        \sigma \setminus \cycleof{\sigma}{N}
    }
}
\int
\prod_{i=1}^{N - q}
\dd{
    \posbead{i}{1}
}
\func{I}{
    \posbead{i}{1},
    \posbead{
        \rep{
            \sigma \setminus \cycleof{\sigma}{N}
        }(
            i
        )
    }{
        1
    }
}
=
\left[
    Z_{P, N-q}
    \fact{\left(N - q\right)}
    \left(
        \frac{
            \pi    
        }{
            a
        }
    \right)^{
        \frac{
            d \left(N - q\right) P
        }{
            2
        }
    }
\right]
\cdot
\binom{N - 1}{q - 1}
\cdot
\fact{\left(q - 1\right)}.
\end{equation*}
The term in the square brackets follows from the definition of $Z_{P,N}$. The terms multiplying the brackets are combinatorial in nature, as they represent the number of possible different cycles of length $q$ that contain $N$. Substituting this into the previous equation yields the desired recurrence.
\end{proof}

\ifmain
\else
  \bibliography{rsc}
  \end{document}
\fi
\fi

\end{document}